\documentclass[a4paper,12pt]{amsart}
\usepackage{amsmath,amssymb,amsthm}
\usepackage{amscd}
\usepackage{mathrsfs}
\usepackage[usenames,dvipsnames]{color}
\usepackage{graphicx}
\usepackage[all]{xy}
\usepackage{eucal}
\usepackage{extarrows} 
\usepackage{tabularx}
\usepackage{tikz} 
\usepackage{xparse} 
\usepackage{colordvi}
\usepackage{multicol}
\usepackage{enumitem}
\usepackage[normalem]{ulem}
\usepackage{epsfig}
 \usepackage{ifpdf}
  \ifpdf
   \usepackage[colorlinks,final,backref=page,hyperindex]{hyperref}
  \else
   \usepackage[colorlinks,final,backref=page,hyperindex,hypertex]{hyperref}
  \fi

\usepackage{xspace}

\usetikzlibrary{arrows}  
\usetikzlibrary{decorations.pathmorphing}  

\begin{document}

\newcommand\cutoffint{\mathop{-\hskip -4mm\int}\limits}
\newcommand\cutoffsum{\mathop{-\hskip -4mm\sum}\limits}
\newcommand\cutoffzeta{-\hskip -1.7mm\zeta} 
\newcommand{\goth}[1]{\ensuremath{\mathfrak{#1}}}
\newcommand{\bbox}{\normalsize {}%
        \nolinebreak \hfill $\blacksquare$ \medbreak \par}
\newcommand{\simall}[2]{\underset{#1\rightarrow#2}{\sim}}

\newtheorem{theorem}{Theorem}[section]
\newtheorem{prop}[theorem]{Proposition}
\newtheorem{lemdefn}[theorem]{Lemma-Definition}

\newtheorem{propdefn}[theorem]{Proposition-Definition}
\newtheorem{lem}[theorem]{Lemma}
\newtheorem{thm}[theorem]{Theorem}
\newtheorem{coro}[theorem]{Corollary}
\newtheorem{claim}[theorem]{Claim}
\theoremstyle{definition}
\newtheorem{defn}[theorem]{Definition}
\newtheorem{rk}[theorem]{Remark}
\newtheorem{ex}[theorem]{Example}
\newtheorem{coex}[theorem]{Counterexample}

\renewcommand{\theenumi}{{\it\roman{enumi}}}
\renewcommand{\theenumii}{\alpha{enumii}}

\newenvironment{thmenumerate}{\leavevmode\begin{enumerate}[leftmargin=1.5em]}{\end{enumerate}}

\newcommand{\delete}[1]{{}}
\newcommand{\optional}[1]{{\color{lightgray}  #1}}

\newcommand{\nc}{\newcommand}

\nc{\dforest}{\mathfrak{f}}
\nc{\mlabel}[1]{\label{#1}}  
\nc{\mcite}[1]{\cite{#1}}  
\nc{\mref}[1]{\ref{#1}}  
\nc{\mbibitem}[1]{\bibitem{#1}} 

\delete{
\nc{\mlabel}[1]{\label{#1}  
{\hfill \hspace{1cm}{\small{{\ }\hfill(#1)}}}}
\nc{\mcite}[1]{\cite{#1}{\small{{{\ }(#1)}}}}  
\nc{\mref}[1]{\ref{#1}{{{{\ }(#1)}}}}  
\nc{\mbibitem}[1]{\bibitem[\bf #1]{#1}} 
}



\newcommand{\LTwoLadder}[2]{\begin{picture}(12,5)(0,-1)
\put(3,-2){\circle*{2}}
\put(3,-2){\line(0,1){7}}
\put(3,6){\circle*{2}}
\put(4,-4){\tiny #1}
\put(4,4){\tiny #2}
\end{picture}}

\nc{\ola}[1]{\stackrel{#1}{\longrightarrow}}
\nc{\mtop}{\top\hspace{-1mm}}
\nc{\mrm}[1]{{\rm #1}}
\nc{\depth}{{\mrm d}}
\nc{\id}{\mrm{id}}
\nc{\Id}{\mathrm{Id}}
\nc{\mapped}{operated\xspace}
\nc{\Mapped}{Operated\xspace}
\newcommand{\redtext}[1]{{\textcolor{red}{#1}}}
\newcommand{\Hol}{\text{Hol}}
\newcommand{\Mer}{\text{Mer}}
\newcommand{\n}{\text{lin}}
\nc{\ot}{\otimes}
\nc{\mphi}{\eta}
\nc{\Hom}{\mathrm{Hom}}
\nc{\CS}{\mathcal{CS}}
\nc{\bfk}{\mathbf{K}}
\nc{\lwords}{\calw}
\nc{\ltrees}{\calf}
\nc{\lpltrees}{\calp}
\nc{\Map}{\mathrm{Map}}
\nc{\rep}{\beta}
\nc{\free}[1]{\bar{#1}}
\nc{\OS}{\mathbf{OS}}
\nc{\OM}{\mathbf{OM}}
\nc{\OA}{\mathbf{OA}}
\nc{\based}{based\xspace}
\nc{\tforall}{\text{ for all }}
\nc{\hwp}{\widehat{P}^\calw}
\nc{\sha}{{\mbox{\cyr X}}}
\font\cyr=wncyr10 \font\cyrs=wncyr7
\nc{\Mor}{\mathrm{Mor}}
\def\lc{\lfloor}
\def\rc{\rfloor}
\nc{\oF}{{\overline{F}}}
\nc{\mge}{_{bu}\!\!\!\!{}}
\newcommand{\bottop}{\top\hspace{-0.8em}\bot}

\newcommand{\W}{\mathbb{W}}
\newcommand{\R}{\mathbb{R}}
\newcommand{\bbR}{\mathbb{R}}
\newcommand{\bbC}{\mathbb{C}}

\newcommand{\C}{\mathbb{C}}
\newcommand{\K}{\mathbb{K}}
\newcommand{\Z}{\mathbb{Z}}
\newcommand{\Q}{\mathbb{Q}}
\newcommand{\bbB}{\mathbb{B}}
\newcommand{\N}{\mathbb{N}}
\newcommand{\F}{\mathbb{F}}
\newcommand{\T}{\mathbb{T}}
\newcommand{\bbG}{\mathbb{G}}
\newcommand{\U}{\mathbb{U}}
\newcommand{\loc}{locality\xspace}
\newcommand{\Loc}{Locality\xspace}
\newcommand {\frakc}{{\mathfrak {c}}}
\newcommand {\frakd}{{\mathfrak {d}}}
\newcommand {\frakr}{{\mathfrak {r}}}
\newcommand {\fraku}{{\mathfrak {u}}}
\newcommand {\fraks}{{\mathfrak {s}}}
\newcommand{\frakS}{S}
\newcommand {\bbf}{\ltrees}
\newcommand {\bbg}{{\mathbb{G}}}
\newcommand {\bbp}{\lpltrees}
\newcommand {\bbw}{{\mathbb{W}}}
\newcommand {\cala}{{\mathcal {A}}}
\newcommand {\calb}{\mathcal {B}}
\newcommand {\calc}{{\mathcal {C}}}
\newcommand {\cald}{{\mathcal {D}}}
\newcommand {\cale}{{\mathcal {E}}}
\newcommand {\calf}{{\mathcal {F}}}
\newcommand {\calg}{{\mathcal {G}}}
\newcommand {\calh}{\mathcal{H}}
\newcommand {\cali}{\mathcal{I}}
\newcommand {\call}{{\mathcal {L}}}
\newcommand {\calm}{{\mathcal {M}}}
\newcommand {\calp}{{\mathcal {P}}}
\newcommand {\calr}{{\mathcal {R}}}
\newcommand {\cals}{{\mathcal {S}}}
\newcommand {\calt}{{\mathcal {T}}}
\newcommand {\calv}{{\mathcal {V}}}
\newcommand {\calw}{{\mathcal {W}}}
\newcommand {\calms}{{\mathcal {M}}{\mathcal {S}}}
\nc{\vep}{\varepsilon}
\def \e {{\epsilon}}
\nc {\diagtree}{{\bf D}}
\newcommand{\sy  }[1]{{\color{purple}  #1}} 
\newcommand{\cy}[1]{{\color{cyan}  #1}}
\newcommand{\cyt}[1]{{\color{cyan}\texttt{  #1}}}
\newcommand{\zb }[1]{{\color{blue}  #1}}
\newcommand{\li}[1]{{\color{red} #1}}
\newcommand{\lir}[1]{{\it\color{red} (Li: #1)}}
\newcommand{\lit}[2]{\sout{\color{red}{#1}}{\color{red} #2}}


\newcommand{\tdun}[1]{\begin{picture}(10,5)(-2,-1)
\put(0,0){\circle*{2}}
\put(3,-2){\tiny #1}
\end{picture}}

\newcommand{\tddeux}[2]{\begin{picture}(12,5)(0,-1)
\put(3,0){\circle*{2}}
\put(3,0){\line(0,1){5}}
\put(3,5){\circle*{2}}
\put(3,-2){\tiny #1}
\put(3,4){\tiny #2}
\end{picture}}

\newcommand{\tdtroisun}[3]{\begin{picture}(20,12)(-5,-1)
\put(3,0){\circle*{2}}
\put(-0.65,0){$\vee$}
\put(6,7){\circle*{2}}
\put(0,7){\circle*{2}}
\put(5,-2){\tiny #1}
\put(6,5){\tiny #2}
\put(-5,8){\tiny #3}
\end{picture}}

\def\ta1{{\scalebox{0.25}{ 
\begin{picture}(12,12)(38,-38)
\SetWidth{0.5} \SetColor{Black} \Vertex(45,-33){5.66}
\end{picture}}}}

\def\tb2{{\scalebox{0.25}{ 
\begin{picture}(12,42)(38,-38)
\SetWidth{0.5} \Vertex(45,-3){5.66}
\SetWidth{1.0} \Line(45,-3)(45,-33) \SetWidth{0.5}
\Vertex(45,-33){5.66}
\end{picture}}}}

\def\tc3{{\scalebox{0.25}{ 
\begin{picture}(12,72)(38,-38)
\SetWidth{0.5} \SetColor{Black} \Vertex(45,27){5.66}
\SetWidth{1.0} \Line(45,27)(45,-3) \SetWidth{0.5}
\Vertex(45,-33){5.66} \SetWidth{1.0} \Line(45,-3)(45,-33)
\SetWidth{0.5} \Vertex(45,-3){5.66}
\end{picture}}}}

\def\td31{{\scalebox{0.25}{ 
\begin{picture}(42,42)(23,-38)
\SetWidth{0.5} \SetColor{Black} \Vertex(45,-3){5.66}
\Vertex(30,-33){5.66} \Vertex(60,-33){5.66} \SetWidth{1.0}
\Line(45,-3)(30,-33) \Line(60,-33)(45,-3)
\end{picture}}}}

\def\xtd31{{\scalebox{0.35}{ 
\begin{picture}(70,42)(13,-35)
\SetWidth{0.5} \SetColor{Black} \Vertex(45,-3){5.66}
\Vertex(30,-33){5.66} \Vertex(60,-33){5.66} \SetWidth{1.0}
\Line(45,-3)(30,-33) \Line(60,-33)(45,-3)
\put(38,-38){\em \huge x}
\end{picture}}}}

\def\ytd31{{\scalebox{0.35}{ 
\begin{picture}(70,42)(13,-35)
\SetWidth{0.5} \SetColor{Black} \Vertex(45,-3){5.66}
\Vertex(30,-33){5.66} \Vertex(60,-33){5.66} \SetWidth{1.0}
\Line(45,-3)(30,-33) \Line(60,-33)(45,-3)
\put(38,-38){\em \huge y}
\end{picture}}}}

\def\xldec41r{{\scalebox{0.35}{ 
\begin{picture}(70,42)(13,-45)
\SetColor{Black}
\SetWidth{0.5} \Vertex(45,-3){5.66}
\Vertex(30,-33){5.66} \Vertex(60,-33){5.66}
\Vertex(60,-63){5.66}
\SetWidth{1.0}
\Line(45,-3)(30,-33) \Line(60,-33)(45,-3)
\Line(60,-33)(60,-63)
\put(38,-38){\em \huge x}

\end{picture}}}}

\def\xyrlong{{\scalebox{0.35}{ 
\begin{picture}(70,72)(13,-48)
\SetColor{Black}
\SetWidth{0.5} \Vertex(45,-3){5.66}
\Vertex(30,-33){5.66} \Vertex(60,-33){5.66} \SetWidth{1.0}
\Line(45,-3)(30,-33) \Line(60,-33)(45,-3)
\put(38,-38){\em\huge x}
\SetWidth{0.5}
\Vertex(45,-63){5.66} \Vertex(75,-63){5.66} \SetWidth{1.0}
\Line(60,-33)(45,-63) \Line(60,-33)(75,-63)
\put(55,-63){\em\huge y}
\end{picture}}}}

\def\xyllong{{\scalebox{0.35}{ 
\begin{picture}(70,72)(13,-48)
\SetColor{Black}
\SetWidth{0.5} \Vertex(45,-3){5.66}
\Vertex(30,-33){5.66} \Vertex(60,-33){5.66} \SetWidth{1.0}
\Line(45,-3)(30,-33) \Line(60,-33)(45,-3)
\put(40,-33){\em\huge y}
\SetWidth{0.5}
\Vertex(15,-63){5.66} \Vertex(45,-63){5.66} \SetWidth{1.0}
\Line(30,-33)(15,-63) \Line(30,-33)(45,-63)
\put(25,-63){\em\huge x}
\end{picture}}}}

\def\xyldec43{{\scalebox{0.35}{ 
\begin{picture}(70,62)(13,-25)
\SetColor{Black}
\SetWidth{0.5} \Vertex(45,-3){5.66}
\Vertex(15,-33){5.66} \Vertex(45,-38){5.66}
\Vertex(75,-33){5.66}
\SetWidth{1.0}
\Line(45,-3)(15,-33) \Line(45,-3)(45,-38)
\Line(45,-3)(74,-33)
\put(25,-33){\em\huge x}
\put(50,-33){\em\huge y}
\end{picture}}}}

\def\te4{{\scalebox{0.25}{ 
\begin{picture}(12,102)(38,-8)
\SetWidth{0.5} \SetColor{Black} \Vertex(45,57){5.66}
\Vertex(45,-3){5.66} \Vertex(45,27){5.66} \Vertex(45,87){5.66}
\SetWidth{1.0} \Line(45,57)(45,27) \Line(45,-3)(45,27)
\Line(45,57)(45,87)
\end{picture}}}}

\def\tf41{{\scalebox{0.25}{ 
\begin{picture}(42,72)(38,-8)
\SetWidth{0.5} \SetColor{Black} \Vertex(45,27){5.66}
\Vertex(45,-3){5.66} \SetWidth{1.0} \Line(45,27)(45,-3)
\SetWidth{0.5} \Vertex(60,57){5.66} \SetWidth{1.0}
\Line(45,27)(60,57) \SetWidth{0.5} \Vertex(75,27){5.66}
\SetWidth{1.0} \Line(75,27)(60,57)
\end{picture}}}}

\def\tg42{{\scalebox{0.25}{ 
\begin{picture}(42,72)(8,-8)
\SetWidth{0.5} \SetColor{Black} \Vertex(45,27){5.66}
\Vertex(45,-3){5.66} \SetWidth{1.0} \Line(45,27)(45,-3)
\SetWidth{0.5} \Vertex(15,27){5.66} \Vertex(30,57){5.66}
\SetWidth{1.0} \Line(15,27)(30,57) \Line(45,27)(30,57)
\end{picture}}}}

\def\th43{{\scalebox{0.25}{ 
\begin{picture}(42,42)(8,-8)
\SetWidth{0.5} \SetColor{Black} \Vertex(45,-3){5.66}
\Vertex(15,-3){5.66} \Vertex(30,27){5.66} \SetWidth{1.0}
\Line(15,-3)(30,27) \Line(45,-3)(30,27) \Line(30,27)(30,-3)
\SetWidth{0.5} \Vertex(30,-3){5.66}
\end{picture}}}}

\def\thII43{{\scalebox{0.25}{ 
\begin{picture}(72,57) (68,-128)
    \SetWidth{0.5}
    \SetColor{Black}
    \Vertex(105,-78){5.66}
    \SetWidth{1.5}
    \Line(105,-78)(75,-123)
    \Line(105,-78)(105,-123)
    \Line(105,-78)(135,-123)
    \SetWidth{0.5}
    \Vertex(75,-123){5.66}
    \Vertex(105,-123){5.66}
    \Vertex(135,-123){5.66}
  \end{picture}
  }}}

\def\thj44{{\scalebox{0.25}{ 
\begin{picture}(42,72)(8,-8)
\SetWidth{0.5} \SetColor{Black} \Vertex(30,57){5.66}
\SetWidth{1.0} \Line(30,57)(30,27) \SetWidth{0.5}
\Vertex(30,27){5.66} \SetWidth{1.0} \Line(45,-3)(30,27)
\SetWidth{0.5} \Vertex(45,-3){5.66} \Vertex(15,-3){5.66}
\SetWidth{1.0} \Line(15,-3)(30,27)
\end{picture}}}}

\def\xthj44{{\scalebox{0.35}{ 
\begin{picture}(42,72)(8,-8)
\SetWidth{0.5} \SetColor{Black} \Vertex(30,57){5.66}
\SetWidth{1.0} \Line(30,57)(30,27) \SetWidth{0.5}
\Vertex(30,27){5.66} \SetWidth{1.0} \Line(45,-3)(30,27)
\SetWidth{0.5} \Vertex(45,-3){5.66} \Vertex(15,-3){5.66}
\SetWidth{1.0} \Line(15,-3)(30,27)
\put(25,-3){\em\huge x}
\end{picture}}}}

\def\ti5{{\scalebox{0.25}{ 
\begin{picture}(12,132)(23,-8)
\SetWidth{0.5} \SetColor{Black} \Vertex(30,117){5.66}
\SetWidth{1.0} \Line(30,117)(30,87) \SetWidth{0.5}
\Vertex(30,87){5.66} \Vertex(30,57){5.66} \Vertex(30,27){5.66}
\Vertex(30,-3){5.66} \SetWidth{1.0} \Line(30,-3)(30,27)
\Line(30,27)(30,57) \Line(30,87)(30,57)
\end{picture}}}}

\def\tj51{{\scalebox{0.25}{ 
\begin{picture}(42,102)(53,-38)
\SetWidth{0.5} \SetColor{Black} \Vertex(61,27){4.24}
\SetWidth{1.0} \Line(75,57)(90,27) \Line(60,27)(75,57)
\SetWidth{0.5} \Vertex(90,-3){5.66} \Vertex(60,27){5.66}
\Vertex(75,57){5.66} \Vertex(90,-33){5.66} \SetWidth{1.0}
\Line(90,-33)(90,-3) \Line(90,-3)(90,27) \SetWidth{0.5}
\Vertex(90,27){5.66}
\end{picture}}}}

\def\tk52{{\scalebox{0.25}{ 
\begin{picture}(42,102)(23,-8)
\SetWidth{0.5} \SetColor{Black} \Vertex(60,57){5.66}
\Vertex(45,87){5.66} \SetWidth{1.0} \Line(45,87)(60,57)
\SetWidth{0.5} \Vertex(30,57){5.66} \SetWidth{1.0}
\Line(30,57)(45,87) \SetWidth{0.5} \Vertex(30,-3){5.66}
\SetWidth{1.0} \Line(30,-3)(30,27) \SetWidth{0.5}
\Vertex(30,27){5.66} \SetWidth{1.0} \Line(30,57)(30,27)
\end{picture}}}}

\def\tl53{{\scalebox{0.25}{ 
\begin{picture}(42,102)(8,-8)
\SetWidth{0.5} \SetColor{Black} \Vertex(30,57){5.66}
\Vertex(30,27){5.66} \SetWidth{1.0} \Line(30,57)(30,27)
\SetWidth{0.5} \Vertex(30,87){5.66} \SetWidth{1.0}
\Line(30,27)(45,-3) \SetWidth{0.5} \Vertex(15,-3){5.66}
\SetWidth{1.0} \Line(15,-3)(30,27) \Line(30,57)(30,87)
\SetWidth{0.5} \Vertex(45,-3){5.66}
\end{picture}}}}

\def\tm54{{\scalebox{0.25}{ 
\begin{picture}(42,72)(8,-38)
\SetWidth{0.5} \SetColor{Black} \Vertex(30,-3){5.66}
\SetWidth{1.0} \Line(30,27)(30,-3) \Line(30,-3)(45,-33)
\SetWidth{0.5} \Vertex(15,-33){5.66} \SetWidth{1.0}
\Line(15,-33)(30,-3) \SetWidth{0.5} \Vertex(45,-33){5.66}
\SetWidth{1.0} \Line(30,-33)(30,-3) \SetWidth{0.5}
\Vertex(30,-33){5.66} \Vertex(30,27){5.66}
\end{picture}}}}

\def\tn55{{\scalebox{0.25}{ 
\begin{picture}(42,72)(8,-38)
\SetWidth{0.5} \SetColor{Black} \Vertex(15,-33){5.66}
\Vertex(45,-33){5.66} \Vertex(30,27){5.66} \SetWidth{1.0}
\Line(45,-33)(45,-3) \SetWidth{0.5} \Vertex(45,-3){5.66}
\Vertex(15,-3){5.66} \SetWidth{1.0} \Line(30,27)(45,-3)
\Line(15,-3)(30,27) \Line(15,-3)(15,-33)
\end{picture}}}}

\def\tp56{{\scalebox{0.25}{ 
\begin{picture}(66,111)(0,0)
\SetWidth{0.5} \SetColor{Black} \Vertex(30,66){5.66}
\Vertex(45,36){5.66} \SetWidth{1.0} \Line(30,66)(45,36)
\Line(15,36)(30,66) \SetWidth{0.5} \Vertex(30,6){5.66}
\Vertex(60,6){5.66} \SetWidth{1.0} \Line(60,6)(45,36)
\SetWidth{0.5}
\SetWidth{1.0} \Line(45,36)(30,6) \SetWidth{0.5}
\Vertex(15,36){5.66}
\end{picture}}}}

\def\tq57{{\scalebox{0.25}{ 
\begin{picture}(81,111)(0,0)
\SetWidth{0.5} \SetColor{Black} \Vertex(45,36){5.66}
\Vertex(30,6){5.66} \Vertex(60,6){5.66} \SetWidth{1.0}
\Line(60,6)(45,36) \SetWidth{0.5}
\SetWidth{1.0} \Line(45,36)(30,6) \SetWidth{0.5}
\Vertex(75,36){5.66} \SetWidth{1.0} \Line(45,36)(60,66)
\Line(60,66)(75,36) \SetWidth{0.5} \Vertex(60,66){5.66}
\end{picture}}}}

\def\tr58{{\scalebox{0.25}{ 
\begin{picture}(81,111)(0,0)
\SetWidth{0.5} \SetColor{Black} \Vertex(60,6){5.66}
\Vertex(75,36){5.66} \SetWidth{1.0} \Line(60,66)(75,36)
\SetWidth{0.5} \Vertex(60,66){5.66}
\SetWidth{1.0} \Line(60,36)(60,66) \Line(60,6)(60,36)
\SetWidth{0.5} \Vertex(60,36){5.66} \Vertex(45,36){5.66}
\SetWidth{1.0} \Line(60,66)(45,36)
\end{picture}}}}

\def\ts59{{\scalebox{0.25}{ 
\begin{picture}(81,111)(0,0)
\SetWidth{0.5} \SetColor{Black}
\Vertex(75,36){5.66} \SetWidth{1.0} \Line(60,66)(75,36)
\SetWidth{0.5} \Vertex(60,66){5.66}
\SetWidth{1.0} \Line(60,36)(60,66) \SetWidth{0.5}
\Vertex(60,36){5.66} \Vertex(45,36){5.66} \SetWidth{1.0}
\Line(60,66)(45,36) \Line(75,6)(75,36) \SetWidth{0.5}
\Vertex(75,6){5.66}
\end{picture}}}}

\def\tt591{{\scalebox{0.25}{ 
\begin{picture}(81,111)(0,0)
\SetWidth{0.5} \SetColor{Black}
\Vertex(75,36){5.66} \SetWidth{1.0} \Line(60,66)(75,36)
\SetWidth{0.5} \Vertex(60,66){5.66}
\SetWidth{1.0} \Line(60,36)(60,66) \SetWidth{0.5}
\Vertex(60,36){5.66} \Vertex(45,36){5.66} \SetWidth{1.0}
\Line(60,66)(45,36) \SetWidth{0.5} \Vertex(45,6){5.66}
\SetWidth{1.0} \Line(45,6)(45,36)
\end{picture}}}}

\def\bigdect{{\scalebox{0.4}{ 
\begin{picture}(140,120)(0,-60)
\SetColor{Black}
\SetWidth{0.5} \Vertex(70,60){5.66}
\put(48,60){\em\huge$\alpha$}
\SetWidth{1.0} \Line(70,60)(0,20)
\SetWidth{0.5} \Vertex(0,20){5.66}
\put(-15,25){\em\huge$\beta$}
\SetWidth{1.0} \Line(70,60)(70,20)
\SetWidth{0.5} \Vertex(70,20){5.66}
\put(50,20){\em\huge$e$}
\SetWidth{1.0} \Line(70,60)(140,20)
\SetWidth{0.5} \Vertex(140,20){5.66}
\put(150,25){\em\huge$\delta$}

\SetWidth{1.0} \Line(0,20)(-50,-20)
\SetWidth{0.5} \Vertex(-50,-20){5.66}
\put(-70,-20){\em\huge $a$}
\SetWidth{1.0} \Line(0,20)(0,-20)
\SetWidth{0.5} \Vertex(0,-20){5.66}
\put(-20,-20){\em\huge$\gamma$}
\SetWidth{1.0} \Line(0,20)(50,-20)
\SetWidth{0.5} \Vertex(50,-20){5.66}
\put(50,-38){\em\huge$d$}

\SetWidth{1.0} \Line(0,-20)(-30,-50)
\SetWidth{0.5} \Vertex(-30,-50){5.66}
\put(-45,-68){\em\huge$b$}
\SetWidth{1.0} \Line(0,-20)(30,-50)
\SetWidth{0.5} \Vertex(30,-50){5.66}
\put(25,-68){\em\huge$c$}

\SetWidth{1.0} \Line(140,20)(100,-10)
\SetWidth{0.5} \Vertex(100,-10){5.66}
\put(80,-20){\em\huge$f$}
\SetWidth{1.0} \Line(140,20)(140,-20)
\SetWidth{0.5} \Vertex(140,-20){5.66}
\put(150,-30){\em\huge$\sigma$}
\SetWidth{1.0} \Line(140,-20)(140,-60)
\SetWidth{0.5} \Vertex(140,-60){5.66}
\put(150,-70){\em\huge$g$}
\SetWidth{1.0} \Line(140,20)(180,-10)
\SetWidth{0.5} \Vertex(180,-10){5.66}
\put(190,-30){\em\huge$\tau$}
\SetWidth{1.0} \Line(180,-10)(180,-60)
\SetWidth{0.5} \Vertex(180,-60){5.66}
\put(190,-70){\em\huge$h$}
\end{picture}}}}

\title[Locality and renormalisation]{ Locality and renormalisation:\\  universal properties
and integrals on trees}

\author{Pierre Clavier}
\address{Institute of Mathematics,
University of Potsdam,
D-14476 Potsdam, Germany}
\email{clavier@math.uni-potsdam.de}

\author{Li Guo}
\address{Department of Mathematics and Computer Science,
         Rutgers University,
         Newark, NJ 07102, USA}
\email{liguo@rutgers.edu}

\author{Sylvie Paycha}
\address{Institute of Mathematics,
University of Potsdam,
D-14469 Potsdam, Germany\\ On leave from the Universit\'e Clermont-Auvergne\\
Clermont-Ferrand, France}
\email{paycha@math.uni-potsdam.de}

\author{Bin Zhang}
\address{College of Mathematics,
Sichuan University, Chengdu, 610064, China}
\email{zhangbin@scu.edu.cn}

\date{\today}

\begin{abstract}
The purpose of this paper is to build an algebraic framework suited to regularise branched structures emanating from rooted forests and which   encodes the locality principle. This  is achieved by means of the  universal
properties in the \loc framework of properly decorated rooted forests. These universal properties are  then applied to derive the multivariate regularisation of integrals indexed by rooted forests. We study their renormalisation, along the lines of Kreimer's  toy model for Feynman integrals.
\end{abstract}

\subjclass[2010]{08A55, 16T99, 81T15, 32A20, 52B20}

\keywords{locality, renormalisation, algebraic Birkhoff factorisation, partial algebra, operated algebra, Hopf algebra, Rota-Baxter algebra, symbols, Kreimer's toy model}

\maketitle

\tableofcontents
\vfill\eject \noindent

\section{Introduction}
The purpose of this paper is to apply to the study of concrete models an algebraic formulation of the locality principle developed in~\cite{CGPZ1} in the context of the (co)algebraic approach to perturbative quantum field theory
initiated by Connes and Kreimer ~\cite{CK,K}. The concrete model we focus on is Kreimer's rooted forests toy model, which serves as an experimental model to study renormalisation and related renormalisation groups. Rooted forests are useful to generate branched structures, such as  branched integrals and branched sums, but also geometric rough paths.

To apply the  algebraic formulation of \cite{CGPZ1}, we need to investigate locality structures for rooted forests  and to build   locality algebra homomorphisms which encode the regularisation. We then study the renormalisation of branched integrals, and show an invariance property of renormalised values under a similarity transform of decorated rooted forests. The case of branched sums was studied in ~\cite{CGPZ2}, using some of the results of the present paper.

The paper starts with an abstract algebraic part (Sections~\ref
{sec:loc} and~\ref{sec:univpropwordsforests}), which establishes the universal property of   properly decorated rooted forests, and gives the general  algebraic framework to regularize branched structures. This is  followed by an application (Sections~\ref{sec:tree} and~\ref{sec:renCK}) to  branched integrals  in the context of Kreimer's toy model for Feynman integrals, which we revisit  using  a multivariate renormalisation approach.

So in order to renormalise  these a priori divergent integrals, we use a multivariate locality regularisation derived from the universal property of  properly decorated rooted forests instead of the usual univariate regularisation procedure ~\cite{CK} and the multivariate regularisation we used before~\cite {GPZ4}, with the locality setting playing a key role. This multivariate approach gives rise to renormalised branched integrals which inherit the \loc property from properly decorated rooted forests.

We next review the contents of the paper in some detail. The  tools  underlying the algebraic constructions in  this paper are  {\bf  operated   \loc sets and   algebras} discussed in Section \ref{sec:loc}. They are obtained as a combination  of  the notions of operated sets (resp. algebras) and  \loc sets (resp. algebras). Paragraph \ref{subsec:opstruc} is dedicated to operated structures, leaving \loc aside. Theorem \mref{thm:freeopmonoid} builds from rooted trees   decorated by a set $\Omega$, the free object in the category of $\Omega$-operated monoids. Alongside these results,  in Corollary \mref{co:initop}  we derive anew the fact that the space spanned by rooted forests decorated by $\Omega$ is  the
initial object in the category of $\Omega$-operated monoids~\cite{G1,Mo}.

The universal property is then implemented to lift maps from the decoration set $\Omega$ to maps defined on rooted forests $\calf_\Omega$ decorated by $\Omega$ when $\Omega $ is equipped with additional structures. In Proposition~\mref{prop:lift_phi_non_loc}  a map $\phi:\Omega\longrightarrow\Omega$  on the decorating
monoid (resp. algebra) $\Omega$ is lifted to a morphism  of monoids (resp. algebras) $\widehat\phi:\calf_\Omega\longrightarrow\Omega$ (resp. $\widehat\phi:\bfk \calf_\Omega\longrightarrow\Omega$ for the monoid algebra $\bfk\calf_\Omega$ with coefficients in a field $\bfk$) from rooted forests decorated by $\Omega$. The universal property is further discussed in a relative context in Proposition \mref{pp:liftrel} to lift a morphism $\phi:\Omega_1\longrightarrow\Omega_2$ of monoids (algebras) to a morphism
$\widehat\phi:\calf_{\Omega_1}\longrightarrow\calf_{\Omega_2}$ ($\widehat\phi:\bfk \calf_{\Omega_1}\longrightarrow \bfk\calf_{\Omega_2}$) of monoids (algebras).

In Paragraph \ref{subs:loc}, we recall the concept of locality and in Paragraph \ref{subsection:loc_op_structures} we introduce the notions of operated \loc structures, preceded by the notions of operated \loc sets in Definition~\mref{defn:loc_op_set},  of operated \loc
 semigroups, monoids and algebras.

Section \ref{sec:univpropwordsforests} deals with the universal property of properly decorated rooted forests.  In
Paragraph \ref{subsec:propdecforests}  we consider properly decorated rooted forests, to which we extend the universal properties of ordinary decorated rooted forests incorporating \loc.
 Theorem \mref{thm:univ_prop_trees_loc} is the \loc version of Corollary \mref{co:initop}, while Proposition~\ref{pp:liftedphi} and Corollary~\ref{thm:existencemapbranching} are the \loc version of Propositions~\mref{pp:liftrel}~and \mref{prop:lift_phi_non_loc} respectively. Corollary \ref{thm:existencemapbranching} was used in \cite{CGPZ2} to renormalised discrete sums attached to trees.

In Section \mref{sec:tree} we implement the {aforementioned} universal \loc properties on forests to the study of branched integrals in the context of Kreimer's toy model \cite{K}, to construct a regularisation of branched integrals in the locality setup.

For this purpose, we build  a locality morphism in several steps.
In Paragraph \ref{ss:ComplexPower}, we first introduce the  group ring $\calm[\call]$ over the algebra $\calm$ of meromorphic germs with linear poles generated by the additive monoid $\call$ of multivariate linear forms,  and equip it  with a locality algebra structure $\perp$. In  Paragraph \ref{subsec:localg}, we   view $(\calm[\call], \perp)$  as a locality algebra operated by $\call$ (Lemma \ref{lem:opL}).
Paragraph \ref{subsec:locmorforests} is dedicated to locality morphisms defined on forests decorated   by $(\call,\perp)$. The universal property  discussed in Theorem~\mref{thm:locinit} yields a $\calm[\call]$-valued locality algebra homomorphism ${\mathcal R}$ (Lemma \ref{lem:RFd}, Lemma \ref{lem:Rformula}) defined on $\call$-decorated forests, from which we build an $\call$-valued locality morphism ${ \mathcal  R}_1 $ on $\call$-decorated rooted forests. We moreover provide an explicit formula of its evaluation on properly decorated forests and show that it is a \loc algebra homomorphism in Proposition~\mref{thm:locrenorm}.

In Section \ref{sec:renCK}, following the renormalisation scheme by locality morphisms presented in Paragraph~\ref{subsec:renmorph}, we build  a renormalised map $ \pi_+\,  { \mathcal  R}_1$. It takes values in holomorphic germs
at zero which we evaluate   at zero  to build the renormalised map as a  \loc character  $\calr ^{\rm ren}$ on the \loc algebra of properly decorated rooted forests (Proposition-Definition \mref{thm:locrenormr}). An interesting
feature of this renormalisation process is that similar properly decorated rooted forests (Definition \ref{defn:similar}) have the same renormalised values, as shown in Paragraph~\mref{ss:similar}. This results from
Theorem~\ref {thm:reductionproc}  which  yields an algorithm to evaluate the renormalised value of any given branched integral and whose proof uses  computational techniques for multivariate meromorphic germs with linear poles developed in  \cite{GPZ3}.

In conclusion, this paper aims at the application of our locality principle, provides  new algebraic tools for multivariate regularisation and renormalisation associated with rooted tree structures, some of which were used in \cite{CGPZ2}. This multivariate renormalisation scheme is then implemented on a non-trivial example, namely Kreimer's toy model.

\section{\Loc operated sets and algebras}
\mlabel{sec:loc}

In this section we introduce the concepts of \loc operated set, semigroup and algebra, and construct the free objects in the corresponding categories. For this purpose we first revisit these concepts and constructions without the \loc conditions, and later adapt to the \loc setting.

\subsection{Operated structures and free objects}
\label{subsec:opstruc}
After recalling the concepts of operated structures, we  give the free construction for operated sets, operated semigroups and monoids, and operated algebras, successively.

\subsubsection{The notion of operated structures}

Let $\Omega$ be a set. Recall that an {\bf $\Omega$-operated set (resp. semigroup, monoid, vector space, unital algebra)} $(U,\beta)$
(resp. $(U,m_U,\beta)$, $(U,m_U,1_U,\beta)$, $(U,+,\beta)$, $(U,m_U,1_U, +,\beta)$)
(see \cite{G1, G2}) is a  set (resp. semigroup, monoid, vector space, unital algebra) $U$ together with a set of operators
$$\beta:=\beta_U^\Omega:=\{\beta^\omega:=\beta_U^\omega:U\to U\,|\, \omega\in \Omega\}$$
parameterised by $\Omega$. More precisely, this means that there is a map
\begin{equation}
\beta_U=\beta_U^\Omega:	\Omega\times U\longrightarrow  U, \quad
	(\omega, u)\longmapsto  \beta_U^\omega(u).
\mlabel{eq:alpha}
\end{equation}
 The maps   $\beta_U^\omega$ are often called {\bf grafting operators}.
In the case of a vector space or a unital algebra, we also assume that the operators $\beta_U^\omega$ are linear.

A homomorphism from an $\Omega$-operated object $(U,\rep_U)$ to an $\Omega$-operated object $(V,\rep_V)$ is a morphism $f:U\to V$ in the corresponding category without the $\Omega$-actions with the property
\begin{equation}
f(\rep_U^\omega(u)) = \rep_V^\omega (f(u))\quad \text{ for all } u\in U, \omega\in \Omega.
\label{eq:opcategory}
\end{equation}
We therefore have the category $\OS_{\Omega}$ (resp ${\bf OSG}_\Omega$, resp. $\OM_\Omega$, resp. $\OA_\Omega$) of $\Omega$-operated sets (resp. semigroups, resp. monoids, resp. algebras)

\subsubsection{Free operated monoids and algebras}
Let us consider free objects in the categories of various operated algebraic structures. For ``classical" algebraic structures without operators, such as associative and Lie algebras, the free objects have a generating set $X$.
For operated algebraic structures, we already have a set of operators. So we need to be careful in distinguishing the two sets: the set $\Omega$ of operators and the set $X$ of generators for a free object.
Even though for the applications in this paper, the generating set $X$ will be taken to be the empty set, we give a uniform approach with arbitrary generating sets which might be applied to broader contexts.

We next construct free objects in the category of $\Omega$-operated monoids and algebras, with the latter following naturally from the former.

Let us introduce some terminology. A {\bf rooted tree} (resp. {\bf planar rooted tree}) is a connected loopless graph (resp. planar graph), whose edges are oriented, thereby equipping the tree with a  partial order on its set of vertices with a unique minimal element, called the {\bf root}. A {\bf rooted forest} (resp. {\bf planar rooted forest}) is a concatenation of (resp. planar) rooted trees.
 Any  maximal vertex (i.e. one that has no element above it) for this partial order is called a {\bf leaf}. The set of leaves of a forest $F$ is denoted by $l(F)$.
For a rooted forest or planar rooted forest $F$, let $V(F)$ denote the set of vertices of $F$. A vertex of $F$ is called {\bf non-leaf} or {\bf interior}  if it is not a leaf vertex. For the tree with unique vertex, the vertex is taken to be a leaf vertex.

The regularised integrals parametrised by trees which arise in the renormalisation procedure considered later in this paper obey the commutativity property. Consequently, we will focus on (non planar) trees and take the forest concatenation to be commutative.

The following concept is a natural generalisation of notions from~\cite{G1,ZGG}. See also~\cite{GZ} for the planar case.

\begin{defn}
Let  $\Omega$ and $X$ be disjoint sets.
An  $(\Omega,X)$-{\bf decorated rooted forest} is a pair $(F, d)$, where $F$ is a rooted forest and $d: V(F)\to \Omega\cup X$  from the set $V(F)$ of vertices of $F$  {is such that $d(V(F)\setminus l(F))\subseteq\Omega$, i.e.}
whose restriction to the non-leaf vertices is in $\Omega$
(but whose restriction to the leaf vertices is in $\Omega\cup X$). Let $\calf _{\Omega,X}$ denote the set of $(\Omega,X)$-decorated rooted forests together with the ``empty tree" denoted $1$.
\mlabel{de:forest}
\end{defn}

For $\omega\in \Omega$, we define the grafting operator
$$B_+^\omega:\calf _{\Omega,X}\to \calf_{\Omega,X}$$
which sends any rooted forest $(F,d)$ to a rooted tree by adding to $(F,d)$ a new root decorated by $\omega$, and sends the empty tree $1$ to the tree
$\bullet_\omega$ with a single leave decorated by $\omega$. We set $B:=\{B_+^\omega\,|\,\omega\in \Omega\}$.

The number of vertices of a rooted forest, which we call the {\bf degree} of the rooted forest, provides a grading on  forests and decorated forests.
The following simple result  further follows from the analog statement in the undecorated case. See~\cite{G1} for example.

\begin{lem}
Let $i:X\to \calf_{\Omega,X}, x\mapsto \bullet_x, x\in X,$ be the canonical embedding of $X$ into $\calf_{\Omega,X}$.
An $(\Omega,X)$-decorated rooted forest $(F,d)$ is either $1$ or can be written in  an unique way as follows:
\begin{enumerate}
 {\item $F\in Im(i)$, that is $F=\bullet_x$ for some $x\in X$.}
\item	
If $F$ is a non-empty tree  {not in the image of $i$}, then $ (F,d)= B_+^{\omega}(\overline{F},\bar{d})$ for some $(\Omega,X)$-decorated rooted forest $(\overline{F},\bar{d})$ $($which might be $1$$)$ with $\deg(F,d)=\deg(\overline{F},\overline{d})+1$;
\item
If $F$ is not a tree, then $(F,d)=(F_1,d_1)\cdots (F_k,d_k), \ k\ge 2,$ with $(\Omega,X)$-rooted trees $F_i\not=1$ and with
$\deg(F,d)=\deg(F_1,d_1)+\cdots+\deg(F_n,d_n)$.
\end{enumerate}
\mlabel{lem:forestdecom}
\end{lem}

The following results for planar rooted forests have been obtained in~\cite{ZGG} when $|X|=1$ and in~\cite{GZ} for general $X$. The results here, for rooted forests, follows from the same argument.

\begin{thm} \mlabel{thm:obj_free_non_loc}
Let $\Omega$ and $X$ be sets with $\Omega$ non-empty. Let $\bfk$ be a field.
\begin{enumerate}
\item
The operated set $(\calf_{\Omega,X},B)$, with the forest concatenation, is an $\Omega$-operated {commutative} monoid with unit $1$. The linear span $(\bfk\calf_{\Omega,X},B)$ is an $\Omega$-operated unital {commutative} $\bfk$-algebra.
\mlabel{it:opm}
\item
The operated monoid $\calf_{\Omega,X}$ together with the map
$$ i:X\to \calf_{\Omega,X}, x\mapsto \bullet_x, x\in X,$$
is the free object on $X$ in the category of $\Omega$-operated commutative monoids.
More precisely, for any $\Omega$-operated commutative monoid $(U,\beta)$ and map $f: X\to U$, there is unique morphism $\free{f}: \calf_{\Omega,X} \to U$ of $\Omega$-operated commutative monoids such that $\free{f}\circ i=f$.
\mlabel{it:freeopm}
\item
The operated algebra $\bfk\,\calf_{\Omega,X}$ is the free object on $X$ in the category of $\Omega$-operated commutative algebras, characterised by a universal property similar to the previous one.
\mlabel{it:freeopa}
\end{enumerate}
\mlabel{thm:freeopmonoid}
\end{thm}

\subsubsection{Initial objects and the relative extensions}

We will be most interested in a special case of $(\Omega,X)$-decorated forests, namely when $X=\emptyset$ is the empty set. Since there is unique map  {(known as the empty map)} from $\emptyset$ to any set,   the subsequent
statement follows directly from
Theorem~\mref{thm:freeopmonoid}.

\begin{coro} Let a set $\Omega$ be given.
\begin{enumerate}
\item
The set $\calf_\Omega:=\calf_{\Omega,\emptyset}$ of $\Omega$-decorated rooted forests with all vertices decorated by $\Omega$ is the initial object in the category of $\Omega$-operated commutative monoids.
\mlabel{it:initmon}
\item
The space $\bfk\,\calf_\Omega$ of $\Omega$-decorated rooted forests is the initial object in the category of $\Omega$-operated commutative  {algebras}.
\mlabel{it:initalg}
\end{enumerate}
\mlabel{co:initop}
\end{coro}
\begin{rk}
 This result was already shown in aforementioned papers~\cite{G1,Mo} by recursively constructing the morphism of $\Omega$-operated algebras. In contrast, our
 proof is a simple application of Theorem \mref{thm:obj_free_non_loc}.
\end{rk}

Let $\Omega$ be a monoid (resp. an algebra). Let $\phi:\Omega\to \Omega$ be a map (resp. linear map). Then $\Omega$ becomes an $\Omega$-operated structure with the operators

\begin{equation}
 \beta_\phi^\omega: \Omega \longrightarrow \Omega, \quad
\omega' \longmapsto \phi(\omega \omega'), \quad \omega, \omega'\in \Omega.
\mlabel{eq:selfop}
\end{equation}

It then follows from Corollary~\mref{co:initop} that

\begin{prop} \mlabel{prop:lift_phi_non_loc} Let $\Omega$ be a commutative monoid $($resp. algebra$)$. A  $($resp. linear$)$ map $\phi:\Omega\to \Omega$  induces a unique homomorphism
$$\widehat{\phi}:  \calf_\Omega \longrightarrow \Omega \quad
(\text{resp. }  \bfk \calf_\Omega \longrightarrow \Omega)$$
of commutative monoids $($resp. algebras$)$ such that
\begin{equation}
\widehat{\phi}(B_+^\omega(F))=\phi\big(\omega(\widehat{\phi}(F))\big).
\mlabel{eq:indphi}
\end{equation}
\mlabel{pp:indphi}
\end{prop}

\begin{proof}
As we observed before the proposition, $(\Omega,\beta_\phi:=\{\beta_\phi^\omega\,|\,\omega\in \Omega\})$ is an $\Omega$-operated commutative monoid (resp. algebra). The existence and uniqueness of a  homomorphism
$\widehat{\phi}: \calf_\Omega \to \Omega$ of $\Omega$-operated monoids (resp. algebras) then follows from Corollary~\mref{co:initop}, and Eq.~(\mref{eq:indphi}) boils down to the compatibility condition for the $\Omega$-operations.
\end{proof}

We next extend the universal property of $\calf_\Omega$ to the relative context.
\begin{prop}
Let $\phi:\Omega_1\to \Omega_2$ be a map and let $(U,\beta)$ be an $\Omega_2$-operated commutative monoid.
Then $U$ has an $\Omega_1$-action induced from $\phi$. Further $\phi$ lifts uniquely to a homomorphism of $\Omega_1$-operated monoids as defined in Eq.~\eqref{eq:opcategory}:
$$ \phi^\sharp: \calf_{\Omega_1} \to U.$$
More precisely $\phi^\sharp$ is characterised by the properties
				  \begin{eqnarray}
					 & \phi^\sharp(1) = 1_{U}, \mlabel{eq:relopid} \\
					 &    \phi^\sharp ((F_1,d_1)\cdots (F_n,d_n)) = \phi^\sharp(F_1,d_1)\cdots\phi^\sharp(F_n,d_n),\mlabel{eq:relopprod}\\
					 & \phi^\sharp \left(  B_+^\omega(F,d)\right) =  \beta_{U,+}^{\phi(\omega)}\left(\phi^\sharp(F,d)\right) \mlabel{eq:relopop}
				 \end{eqnarray}
for $\Omega_1$-decorated rooted forests $(F,d), (F_1,d_1), \cdots, (F_n, d_n) $    and $\omega\in \Omega_1$.

The same applies when $\Omega_1, \Omega_2$ are monoids $($resp. algebras$)$, $U$ is an $\Omega_1$-operated monoid $($resp. algebra$)$ and $\phi: \Omega_1\to \Omega_2$ is a map $($resp. linear map$)$.
\mlabel{pp:liftrel}
\end{prop}
Note that the proposition applies when $\Omega_1=\Omega_2$, whether or not $\phi$ is the identity map. The case $\Omega_1=\Omega_2$ and $\phi=\Id$ will be of interest later (Corollary \mref{thm:existencemapbranching}).
\begin{proof}
The map $\phi:\Omega_1\to \Omega_2$ induces an $\Omega_1$-operated monoid structure $\tilde{\beta}:=\{\tilde{\beta}^\omega\,|\, \omega \in \Omega_1\}$ on $U$ by pull-back
$$ \tilde{\beta}^\omega(u):= \beta^{\phi(\omega)}(u), \quad \text{ for all } \omega \in \Omega_1, u\in U.$$
The universal property of $\calf_{\Omega_1}$ in Corollary~\mref{co:initop} then yields a unique homomorphism
$\phi^\sharp: \calf_{\Omega_1} \to U$
of $\Omega_1$-operated commutative monoids as stated in the proposition.
\end{proof}

\subsection{\Loc sets and algebras}
\label{subs:loc}
We first recall the concept of a \loc set introduced in \cite{CGPZ1}.
\begin{defn}
A {\bf \loc set} is a  couple $(X, \top)$ where $X$ is a set and $ \top\subseteq X\times X$ is a binary {\em symmetric} relation on $X$. For $x_1, x_2\in X$,
denote $x_1\top x_2$ if $(x_1,x_2)\in \top$. We also use the alternative notations $X\times_\top X$ and {$X^{_\top 2}$} for $\top$. \end{defn}

In general, for any subset $U\subset X$, let
 			\begin{equation*}
 			U^\top:=\{x\in X\,|\, (x,U)\subseteq \top \}  			 \end{equation*}
denote the {\bf  {polar} subset} of $U$.
For integers $k\geq 2$, denote
$$ X^{_\top k}:=X\times_\top \cdots _\top X:= \{(x_1,\cdots,x_k)\in X^k\,|\, x_i \top x_j \text{ for all } 1\leq i<j\leq k \}.$$

 We call two subsets $A$ and $B$   of a \loc set $(X,\top )$  {\bf independent} if $A\times B\subset \top.$
Thus a \loc relation $\top$ on a set $X$ induces  a relation on the power set $\mathcal{P}(X)$, which we denote by the same symbol $\top$

Let $(X,\top )$ be a \loc set. Then the relation \loc induced on $\mathcal{P}(X)$ is symmetric, so that $(\mathcal{P}(X),\top ) $ is a \loc set~\mcite{CGPZ1}. Furthermore,
  $\mathcal{P}(X  )^{\top }=\mathcal{P}(X ^{\top })$, as can be checked directly.

Recall that two maps  $\Phi,\Psi:\left( X,\top _X\right)\to \left(Y, \top_Y\right)$ are called {\bf independent}  and we write $\Phi\top \Psi$ if $(\Phi\times \Psi)(\top _X) \subseteq \top _Y$, that is,
${x _1\top _X x _2}$ implies $\Phi(x_1)\top _Y\Psi\left(  x_2\right)$ for $x_1,x_2\in X$.
A map  $\Phi:(X,\top _X) \longrightarrow (Y,\top _Y )$   is called a {\bf \loc map} if $\Phi\top\Phi$. Given two \loc sets $(X,\top_X)$ and $(Y,\top_Y)$, let  $\calm or_\top(X,Y)$ denote the set of
			\loc maps from $X$ to $Y$.

Here are some examples of \loc sets used later on.

\begin{ex}
\begin{enumerate}
\item
 		 	The {\bf power set} ${\mathcal P}(S)$ of any set $S$  can be equipped with the independence relation:
 		 	\begin{equation*}
 		 	A \top B\Longleftrightarrow A\cap B=\emptyset,
 		 	\end{equation*}
 		 	so that $\left({\mathcal P}(S),\top\right)$ is a \loc set with ${\mathcal P}(S)^{\top}=\{\emptyset\}$.
\mlabel{ex:powerset}
\item A \loc structure on decorated forests can be deduced from a \loc structure on the set of decorations.
Indeed, given a \loc set $(\Omega,\, \top _\Omega )$, the set $\calf _{\Omega}$ of ${\Omega}$-decorated rooted forests can be equipped with the following
independence relation induced by that of $\calp (\Omega)$:
\begin{equation}
 					(F_1, d _1)\, \top _{\calf_\Omega}\,(F_2, d _2)\Longleftrightarrow d_1(V(F_1))\, \top _\Omega\, d _2(V(F_2))
 \mlabel{eq:IndForestss}
 \end{equation}
 				Then $(\calf _\Omega,\top _{\calf_\Omega})$ is a \loc set. Let $\bfk\calf _\Omega$ be its linear span, with the induced \loc relation denoted by $\top _{\bfk\,\calf _\Omega}$.
\mlabel{ex:DecoratedForest}
\mlabel{ex:DecoratedPlanarForest}
\end{enumerate}
\mlabel{ex:locsets}
\end{ex}

We also recall the concepts of \loc monoids and \loc algebras.

\begin{defn} \mlabel{defn:lsg}
\begin{enumerate}
\item
A {\bf \loc semigroup} is a \loc set $(G,\top)$ together with a product law defined on $\top$:
$$ m_G: G\times_\top G\longrightarrow  G, (x,y)\mapsto  x\cdot y=m_G(x,y), \quad \text{for all } (x,y)\in\top_G $$
for which the product is compatible with the \loc relation on $G$, more precisely
\begin{equation}\tforall U\subseteq G, \quad  m_G((U^\top\times U^\top)\cap\top)\subset U^\top
\mlabel{eq:semigrouploc}
\end{equation}
and such that
\begin{equation}
(x\cdot y) \cdot z = x\cdot (y\cdot z) \text{ for all }(x,y,z)\in G\times_\top G\times_\top G. 	
\mlabel{eq:asso}
\end{equation}
Note that, because of the condition \eqref{eq:semigrouploc}, both sides of Eq.~(\mref{eq:asso}) are well-defined for any triple in the given subset.
\mlabel{it:lsg}
\item
A  \loc semigroup is {\bf commutative} if $m_G(x,y)=m_G(y,x)$ for $(x,y)\in \top$.
\item
A  {\bf \loc   monoid} is a \loc   semigroup $(G,\top, m_G)$ together with a {\bf unit element} $1_G\in G$ given by the defining property
\[\{1_G\}^\top=G\quad \text{ and }\quad m_G(x, 1_G)= m_G(1_G,x)=x\quad \tforall  x\in G.\]
We denote the \loc  monoid by $(G,\top,m_G, 1_G)$.
\mlabel{defn:partial monoid}
\item A {\bf \loc group} is a \loc monoid  $(G,\top, m_G, 1_G)$ equipped with a \loc map
$$
s: G\longrightarrow  G, \quad
g\longmapsto  s(g), \text{ for all } g\in G,
$$
such that $(g, s(g))\in \top$ and $m_G(g, s(g))= m_G(s(g), g)=1_G$ for any $g\in G$.
\mlabel{it:lg}
\item
A {\bf \loc vector space} is a vector space $V$ equipped with a \loc relation $\top$ which is compatible with the linear structure on $V$ in the sense that, for any  subset $X$ of $V$, $X^\top$ is a linear subspace of $V$.
\item
Given two vector spaces $V$ and $W$, we equip their cartesian product $V\times W$ with a locality structure  $\top:=V\times_\top W \subseteq V\times W$ compatible with the vector space structure on the cartesian product, that is, for $X\subseteq V$ and $Y\subseteq W$, the subsets $X^\top$ and $^\top Y$ are subspaces of $W$ and $V$ respectively. A map $f: V\times_\top W \to U$ to a vector space $U$ is called {\bf \loc bilinear} if
$$f(v_1+v_2,w_1)=f(v_1,w_1)+f(v_2,w_1), \quad f(v_1,w_1+w_2)=f(v_1,w_1)+f(v_1,w_2),$$
$$f(kv_1,w_1)=kf(v_1,w_1), \quad
f(v_1,kw_1)=kf(v_1,w_1)$$
for all $v_1,v_2\in V$, $w_1,w_2\in W$ and $k\in  K $ for which all the pairs arising in the above expressions are in $V\times_\top W$.
\item A {\bf nonunitary \loc  algebra} over $K$ is a \loc vector space $(A,\top)$ over $K$ together with a \loc bilinear map
	$$ m_A: A\times_\top A \to A$$ such that
	$(A,\top, m_A)$ is a \loc semigroup.
	\item A {\bf \loc algebra} is a nonunitary \loc algebra $(A,\top, m_A)$ together with a {\bf unit} $1_A:K\to A$ in the sense that
	$(A,\top, m_A, 1_A)$ is a \loc monoid. We shall omit explicitly  mentioning the unit $1_A$ and the product $m_A$ unless doing so generates ambiguity.
\end{enumerate}
\end{defn}

Here is a straightforward consequence of the  above definition.
\begin{lem}
Let $(G,m_G,\top_G)$ be a \loc semigroup. Let $k\geq  2$ and $1\leq i\leq k$. For $(x_1,\cdots,x_k)\in G^{_\top k}$ we have
\begin{enumerate}
\item
$(\id_G^{i-1}\times m_G\times \id_G^{k-i-1})(x_1,\cdots,x_k)\in G^{_\top (k-1)}.$
\item
$(x_1\cdot \ldots \cdot x_i, x_{i+1}\cdot \ldots \cdot x_k)\in G\times_\top G,$.
\end{enumerate}
\mlabel{lem:locprod}
\end{lem}

\begin {ex}
A  central example of a \loc monoid in this paper is that of trees.
 		 Given a \loc set $(\Omega,\, \top _\Omega )$,  the set $  {\calf}_\Omega $  of
 		 $\Omega$-decorated rooted forests can be equipped with	 the independence relation defined in
 		 Example \mref {ex:locsets}.(\mref{ex:DecoratedForest}).
	The concatenation product of forests induces a disjoint product  $(F_1,d_1)\cdot (F_2, d_2)$ of $\Omega$-decorated forests  defined as
	$(F_1\cdot F_2,d_{F_1\cdot F_2})$ with $d_{F_1\cdot F_2}|_{V(F_1)}=d_{1}$ and
	$d_{F_1\cdot F_2}|_{V(F_2)}=d_{2}$.
Then $\left(\calf_\Omega,\top _{\calf _\Omega}, \cdot,1\right)$ is a \loc monoid which induces a \loc algebra structure on  $\bfk\calf_\Omega$.
  \mlabel{ex:locstructtrees}
\end{ex}

See Paragraph~\ref{ss:ComplexPower} for another important example arising from meromorphic germs with linear poles.

\subsection{\Loc operated structures} \label{subsection:loc_op_structures}

We now combine   the \loc and operated structures.

\begin{defn} \mlabel{defn:loc_op_set}
Let $(\Omega, \top)$ be a \loc set. An {\bf $(\Omega,\top)$-operated \loc set} or simply a {\bf \loc operated set} is a \loc set $(U,\top_U)$ together with a {\bf partial action} $\beta$ of $\Omega$ on $U$ on a subset
$\top_{\Omega,U}:=\Omega\times_\top U\subseteq \Omega \times U$

$$ \beta:  \Omega\times_\top U \longrightarrow U, \ (\omega, x)\mapsto \beta ^\omega (x), $$
satisfying the following compatibility conditions.
	   \begin{enumerate}
\item
For
$$\Omega \times_\top U  \times_\top U: = \{(\omega,u,u')\in \Omega\times U\times U\,|\,
 (u,u')\in \top_U, (\omega, u), (\omega , u')\in \Omega \times_\top U\},$$
the map $\beta \times \Id_U: (\Omega\times_\top U) \times U \longrightarrow U\times U$ restricts to
$$\beta \times \Id_U: \Omega \times_\top U  \times_\top U
\longrightarrow  U\times_\top U.$$
 In other words,  {if} $(\omega,u,u')$ lies in $\Omega \times_\top U  \times_\top U$, then  $(\beta^\omega(u),u')$ lies in $\top_U$.
  \item
For
$$\Omega\times_\top \Omega \times_\top U:=\{(\omega,\omega',u)\in \Omega \times \Omega \times U\,|\, (\omega,\omega')\in T_\Omega, (\omega,u), (\omega',u)\in \Omega\times_\top U\},$$
the map $\Id_\Omega \times \beta: \Omega\times (U \times_\top U) \longrightarrow \Omega \times U$  {restricts to}
$$ \Id_\Omega \times \beta: \Omega\times_\top \Omega \times_\top U \longrightarrow \Omega \times_\top U. $$
In other words, if $(\omega,\omega',u')$ lies in $\Omega \times_\top \Omega  \times_\top U$, then  $(\omega,\beta^{\omega'}(u))$ lies in $\Omega\times_\top U$.
	\end{enumerate}
 \mlabel{defn:locopset}
 \end{defn}

There are variations and generalisations of the compatibility conditions, such as the subsequent {  direct consequences of the axioms}.
\begin {lem}
Let $(U,\top_U)$ be an $(\Omega,\top_\Omega)$-operated \loc set. For $m, n\geq 1$, denote
$$\Omega^{_\top m} \times_\top U^{_\top n}: =\left\{ (\omega_1,\cdots,\omega_m,x_1,\cdots,x_n)\in \Omega^m\times U^n\,\left |\,
\begin{array}{l}
(\omega_1,\cdots,\omega_m)\in \Omega^{_\top m}\\
 (x_1,\cdots,x_n)\in U^{_\top n} \\
 \omega_i\top_{\Omega,U}x_j\quad\forall(i,j)\in[m]\times[n]
 \end{array} \right . \right\}.$$
$ (\Id_\Omega^{m-1} \times \beta^{\omega_m} \times \Id_U^{n-1})(\Omega^{_\top m} \times_\top U^{_\top n})$ is contained in $\Omega^{_\top (m-1)}\times_\top U^{_\top n}$.

With a similar notation, we have
${(\beta \times \beta)} (\Omega\times_\top U\times_\top \Omega\times_\top U)\subseteq U\times_\top U.$
\mlabel{lem:itact}
\end{lem}

\begin{defn}
Let $(\Omega , \top )$ be a \loc set.
 \begin{enumerate}
  \item A {\bf \loc $(\Omega,\top)$-operated semigroup} is a  quadruple  $\left(U,\top _U, \beta ,m_U\right)$, where $(U,\top _U,m_U)$ is a \loc semigroup and
  $\left(U,\top _U, \beta \right)$ is a $(\Omega,\top)$-operated \loc set such that
  \begin{equation}
    (\omega, u, u')\in\Omega\times_\top U\times_\top U~\Longrightarrow~(\omega, uu')\in\Omega\times_\top U;
  \end{equation}
 \item A {\bf \loc $(\Omega,\top)$-operated monoid} is a  quintuple  $\left(U,\top _U, \beta ,m_U,1_U\right)$, where $(U,\top _U,m_U,1_U)$ is an \loc monoid
  and $\left(U,\top _U, \beta, m_U \right)$ is a $(\Omega,\top)$-operated \loc semigroup such that $\Omega \times 1_U\subset \Omega \times _\top U$.
    \item A {\bf $(\Omega,\top)$-operated \loc nonunitary algebra } (resp. {\bf $(\Omega,\top)$-operated \loc unitary algebra}) is a  quadruple  $\left(U,\top _U, \beta ,m_U\right)$ (resp. quintuple
    $(U,\top _U, \beta,$ $m_U, 1_U)$) which is a \loc algebra (resp. unitary algebra) and a \loc $(\Omega,\top)$-operated semigroup (resp. monoid), satisfying the additional condition that for any $\omega\in \Omega$, the set $\{\omega\}^{\top_{\Omega,U}}:=\{ u\in U\,|\, \omega\top_{\Omega,U} u \}$ is a subspace of $U$ on which the action of $\omega$ is linear. More precisely, the last condition means that for any $u_1, u_2\in \{\omega\}^{\top_{\Omega,U}}$ and for any $k_1, k_2\in \bfk$, we have $k_1u_1+k_2u_2\in \{\omega\}^{\top_{\Omega,U}}$ and $\beta^\omega(k_1u_1+k_2u_2)=k_1\beta^\omega(u_1)+k_2\beta^\omega(u_2)$
(resp. this condition and $\Omega \times 1_U\subset \Omega \times _\top U$).
\end{enumerate}
In each case, the $(\Omega,\top)$-operated structure is called {\bf commutative} if the corresponding \loc structure is.
\mlabel{defn:basedlocsg}
\end{defn}
\begin{ex}
A \loc \, semigroup $(G,\top,\cdot) $ is a \loc $(G, \top)$-operated commutative semigroup for the action $\alpha: G\times_\top G\to G$ given by the product on $G$.
\end{ex}

Directly from the definition, we have
\begin {lem}
Let $(U,\top_U)$ be an $(\Omega,\top_\Omega)$-operated \loc semigroup. For $i, j\geq 1$, the subset
$ (\Id_\Omega^i \times m_U \times \Id_U^{j-2})(\Omega^{_\top i} \times_\top U^{_\top j})$  is contained in $\Omega^{_\top i}\times_\top U^{_\top {j-1}}$
\mlabel{lem:actmult}
\end{lem}

\begin{lem}
A \loc operated  semigroup $(U,\mtop_U, \beta, m_U)$ expands to a \loc operated nonunitary algebra $(\bfk U, \mtop_{\bfk U}, \beta, m_{\bfk U})$ by linearity. The same holds for a \loc operated monoid and unitary algebra.
\mlabel{lem:loplinear}
\end{lem}
The proof follows from that of the case without the $(\Omega,\top_{\Omega})$-action.

\begin{defn}\label{defn:morphoplocstr}
Given $(\Omega,\top _\Omega)$-operated \loc structures  (sets, semigroups, monoids, nonunitary algebras, algebras) $(U_i, \top _{U_i}, \beta_i)$, $i=1,2,$ a
{\bf morphism of \loc operated  \loc structures}  is a locality morphism (of sets, semigroups, monoids, nonunitary algebras, algebras) $f: U_1\to U_2$ such that
\begin{itemize}
\item
$(\Id_\Omega\times f)(\Omega\times_\top U_1)\subseteq \Omega\times_\top U_2$
and
\item  $f\circ \beta_1^\omega = \beta_2^{\omega} \circ f$ for all $\omega\in \Omega$.
\end{itemize}
\mlabel{defn:morphismoperatedloc}
\end{defn}	

We therefore have the categories of $\OS_{\Omega,\top _\Omega}$ (resp ${\bf OSG}_{\Omega,\top _\Omega}$, resp. $\OM_{\Omega,,\top _\Omega}$, resp. $\OA_{\Omega,\top _\Omega}$) of $(\Omega,\top _\Omega)$-operated sets (resp. semigroups, resp. monoids, resp. algebras).

\section{Universal properties of decorated rooted forests: the \loc version}
\mlabel{sec:univpropwordsforests}
\subsection{Properly decorated forests}
\mlabel{ss:propdecfor}

We first equip properly decorated rooted forests and the resulting linear space with the structures of a \loc $\Omega$-operated commutative monoid and algebra. We then prove their universal properties in the category of \loc operated commutative monoids and algebras.

\begin{defn}
Let $(\Omega,\top _\Omega)$ be
a \loc set. An {\bf $\Omega$-properly decorated rooted forest} is a decorated rooted forest  $\dforest=( F,  d_F )$
whose vertices are decorated by mutually independent elements of $\Omega$. When $\Omega$
is clear from context, we call them properly decorated forests.

Let  $ \calf _{\Omega,\top _\Omega} $ denote the set of  $\Omega$-properly decorated rooted forests, and $\bfk\,\calf  _{\Omega,\top _\Omega} $ be its linear span. The set $ \calf _{\Omega,\top _\Omega} $
inherits the independence relation $\top _{\calf_{\Omega}}$ of $\calf_\Omega$, denoted by $\top _{\calf_{\Omega, \top _\Omega}}$, and $\bfk\,\calf_{\Omega,\top _\Omega} $ inherits the independence relation $\top _{\bfk\,\calf_{\Omega}}$ of $\bfk\,\calf_\Omega$ denoted by $\top _{\bfk\,\calf_{\Omega, \top _\Omega }}$.
\mlabel{defn:properlydecoratedforest}
\end{defn}

It is easy to see that the disjoint union of forests in $\calf_{\Omega}$ defines a \loc monoid structure on
$\calf _{\Omega,\top _\Omega}$, and thus a \loc algebra structure on $\bfk\,\calf _{\Omega , \top _\Omega}$.
This leads  to the following straightforward yet fundamental result.

\begin{prop}
 Let $(\Omega,\top _\Omega)$ be a \loc set.
 Then
\begin{enumerate}
\item
 $(\calf_{\Omega, \top _\Omega},\top _{\calf_{\Omega, \top _\Omega}}, B_+,\cdot,1)$
  is a \loc $(\Omega,\top _\Omega)$-operated commutative monoid;
  \mlabel{it:opal1}
\item
 $(\bfk\,\calf_{\Omega, \top _\Omega},\top _{\bfk\,\calf_{\Omega, \top _\Omega }}, B_+,\cdot,1)$
is a \loc $(\Omega,\top _\Omega)$-operated  commutative algebra.
\mlabel{it:opal2}
 \end{enumerate}
\mlabel{prop:operatedalgebra}
\end{prop}

\begin{proof}
(\mref{it:opal1}) It follows from the definition that $(\calf_{\Omega, \top _\Omega},\top _{\calf_{\Omega, \top _\Omega}},1)$ is a \loc monoid.   The grafting operators on this \loc monoid further satisfy the conditions of a
\loc $(\Omega,\top _\Omega)$-operated monoid;
\smallskip

\noindent
(\mref{it:opal2})  follows from (\mref{it:opal1})  by linear extension.
\end{proof}

\begin{lem}
Let $(\Omega,\top _{\Omega})$ be a \loc set. An ${\Omega}$-properly decorated rooted forest in $\calf _{\Omega,\top _{\Omega}}$ is either the empty tree $1$ or it can be written uniquely in one of the following forms:
\begin{enumerate}
\item
$(F_1,d_1)\cdots (F_n,d_n), \ n\ge 2,$ with rooted trees $F_i\not=1$ such that
$$(F_i,d_i)\top _{\calf _{\Omega,\top _{\Omega}}} (F_j,d_j),\ 1\leq i\not= j\leq n. $$
Furthermore $\deg(F,d)=\deg(F_1,d_1)+\cdots+\deg(F_n,d_n)$;
\item	
$ B_+^{\omega}(F,d)$ for some $(F,d) \in \calf _{\Omega,\top _{\Omega}}$ which is independent of  $\bullet_\omega.$ Furthermore,
$$\deg(B_+^{\omega}(F,d))=\deg (F,d)+1.$$
\end{enumerate}
\mlabel{lem:fdecomp}
\end{lem}
\begin{proof}
The statements hold for any decorated forest without any independence requirement~\cite{G1}. If a decorated forest has independent decorations, then the independence conditions in the statements automatically hold.
\end{proof}			

The following results are also easy to verify.
\begin{lem}
\begin{enumerate}
\item \mlabel{lem:welldefined_point_i}
For rooted forests $(F_1,d_1),\cdots, (F_n,d_n)\in \bbf_{\Omega,\top_\Omega}$, the product forest $(F_1,d_1)\cdots (F_n,d_n)$ lies in $\bbf_{\Omega,\top_\Omega}$ if and only if
$(F_i,d_i)\top _{\calf _{\Omega,\top _{\Omega}}} (F_j,d_j)$ for $1\leq i\neq j\leq n$;
\item
For a decorated rooted forest $(F,d)\in \calf_{\Omega,,\top_\Omega}$, the rooted tree $B_+^{\omega}(F,d)$ lies in $\calf_{\Omega,\top_\Omega}$ if and only if $(F,d)$ is independent of  $\bullet_\omega.$
\end{enumerate}
\mlabel{lem:welldefined}
\end{lem}

The following characterisation of a morphism of operated commutative monoids will later be useful.
\begin{lem}
Let $U$ be an $(\Omega,\top_\Omega)$-operated commutative monoid. A map $\mphi:\calf_{\Omega,\top_\Omega} \to U$ is a morphism of $(\Omega,\top_\Omega)$-operated commutative monoids if and only if
\begin{enumerate}
\item
$\mphi(1) = 1_{U}$;
\mlabel{it:mid}
\item for any $(F,d),(F',d') \in \calf_{\Omega,\top_\Omega}$, if $(F,d) \top _{\bbf_{\Omega,\top _{\Omega}}}(F',d')$, then $\mphi(F,d) \top _U \mphi(F',d')$;
\mlabel{it:mloc}
\item
For any $(T_1,d_1)\cdots (T_n,d_n)\in \calf_{\Omega,\top_\Omega}$ where $(T_1,d_1), \cdots, (T_n,d_n)$ are decorated rooted trees, the equation
$$\mphi ((T_1,d_1)\cdots (T_n,d_n)) = \mphi(T_1,d_1)\cdots\mphi(T_n,d_n)$$
holds;
\mlabel{it:mprod}
\item If $(\omega,(F,d))$ is in $\Omega\times_\top \calf_{\Omega,\top_\Omega}$, then $(\omega, \mphi (F,d))$ is in $\Omega\times _\top U$;
\mlabel{it:mact}
\item
For any $B_+^\omega(F,d)\in \calf_{\Omega,\top_\Omega}$, the equation
$\mphi \left( B_+^\omega(F,d)\right) =  \beta_{U,+}^{\omega}\left(\mphi(F,d)\right)$ holds.   \mlabel{it:mgraft}
 \end{enumerate}
 \mlabel{lem:char}
\end{lem}

\begin{proof}
($\Longrightarrow$). Suppose  that $\mphi:\calf_{\Omega,\top_\Omega}\to U$ is a morphism of operated commutative monoids. Then conditions~(\mref{it:mloc}), (\mref{it:mact}) and (\mref{it:mgraft}) hold by the locality of the map
$\mphi$ and its compatibility with the actions of $(\Omega,\top_\Omega)$. Condition (\mref{it:mid}) is the unitary condition and Condition (\mref{it:mprod}) follows since the concatenation is the product in
$\calf_{\Omega,\top_\Omega}$. \smallskip

\noindent
($\Longleftarrow$). Now suppose that all the conditions are satisfied, so that we only need to verify that $\mphi$ is multiplicative
for any $(F_1,d_1), (F_2,d_2)\in \calf_{\Omega,\top_\Omega}$ with
$(F_1,d_1)\top_{\Omega} (F_2,d_2)$.

By Lemma~\mref{lem:fdecomp}, we have decompositions
$$ (F_i,d_i)=({T}_{i,1},d_{i,1})\cdots ({T}_{i,n_i},d_{i,n_i}), i=1,2,$$
of $(F_i,d_i)$ into rooted trees. Furthermore, by Lemma~\mref{lem:welldefined}, the concatenation $(F_1,d_1)(F_2,d_2)$ is well-defined in $\calf_{\Omega,\top_\Omega}$ and then the decomposition of $(F_1,d_1)(F_2,d_2)$ in Lemma~\mref{lem:fdecomp} is
$$(F_1,d_1)(F_2,d_2)=({T}_{1,1},d_{1,1})\cdots ({T}_{1,n_1},d_{1,n_1}) ({T}_{{2},1},d_{1,1})\cdots ({T}_{{2},n_{2}},d_{{2},n_{2}}).$$
This gives the multiplicativity of $\mphi$:
\begin{eqnarray*}
\mphi((F_1,d_1)(F_2,d_2))
&=&\mphi(({T}_{1,1},d_{1,1})\cdots ({T}_{1,n_1},d_{1,n_1}) ({T}_{{2},1},d_{1,1})\cdots ({T}_{{2},n_{2}},d_{{2},n_{2}}))\\
&=& \mphi({T}_{1,1},d_{1,1})\cdots \mphi({T}_{1,n_1},d_{1,n_1}) \mphi({T}_{{2},1},d_{1,1})\cdots \mphi({T}_{{2},n_{2}},d_{{2},n_{2}}) \\
&=& \mphi(F_1,d_1)\mphi(F_2,d_2),
\end{eqnarray*}
as required.
\end{proof}

\subsection{The universal property of properly decorated rooted forests}
\label{subsec:propdecforests}
In this part, we study the universal properties of properly decorated rooted forests.
It is  well known that the number of vertices in a forest defines a grading which makes $\calf_\Omega$ into a graded monoid, and further $\bfk\,\calf_\Omega$ into a graded algebra. Restricting the grading to properly decorated forests, we obtain \loc graded operated monoids and algebras.

\begin{thm} \mlabel{thm:univ_prop_trees_loc}
Let a \loc set $(\Omega,\top _{\Omega})$ be given.
\begin {enumerate}
\item  The quintuple $(\calf _{\Omega,\top_\Omega},\top _{\calf _{\Omega,\top_\Omega}},  B_+, \cdot, 1 )$ is the initial object in the category of
$(\Omega,\top _{\Omega})$-operated commutative \loc monoids. More precisely, for any $(\Omega,\top_\Omega)$-operated commutative \loc monoid $U:=(U,\top_U,\beta_U,m_U,1_U)$, there is a unique morphism $\mphi:=\mphi_U: \calf_{\Omega,\top_\Omega}\to U$ of $(\Omega,\top_\Omega)$-operated commutative \loc monoids;
\mlabel{it:commonoid}
\item  $(\bfk\,\calf _{\Omega,\top_\Omega},\top _{\bfk\,\calf _{\Omega,\top_\Omega}},  B_+, \cdot, 1 )$ is the initial object in the category of
$(\Omega,\top _{\Omega})$-operated commutative \loc algebras;
\mlabel{it:commalg}
\end{enumerate}
\mlabel{thm:locinit}
\end{thm}

\begin{proof}
(\mref{it:commonoid})
Let an $(\Omega,\top_\Omega)$-operated commutative \loc monoid $(U,\top_U,\beta_U,m_U,1_U)$ be given. We only need to prove that there is a unique morphism of
$(\Omega,\top _{\Omega})$-operated commutative \loc algebras
$$\mphi=\mphi_U: (\calf _{\Omega,\top_\Omega},\top _{\calf _{\Omega,\top_\Omega}},  B_+, \cdot, 1 )\longrightarrow	 (U,\top_U,\beta_U,m_U,1_U).$$	

By Lemma~\mref{lem:char}, we only need to prove that there is a unique map
$\mphi:\calf_{\Omega,\top_\Omega} \to U$ satisfying the conditions (\mref{it:mid}) --- (\mref{it:mgraft}).
For $k\geq 0$, we set
$$\calf_k: =\{ (F,d)\in \calf_{\Omega,\top_\Omega}\, |\, |F|\leq k\}.$$
We will prove by induction on $k\geq 0$ that there is a unique map $\mphi_k: \bbf_k\to U$ fulfilling the five conditions in
Lemma~\mref{lem:char} when $\calf_{\Omega,\top_\Omega}$ is replaced by $\calf_k$, in which case, we call the corresponding conditions
Condition~$(j)_k$ for $j\in \{\mref{it:mid}, \mref{it:mloc}, \mref{it:mprod}, \mref{it:mact}, \mref{it:mgraft}\}$.

When $k=0$, we have $\calf_k=\{1\}$. Then only Condition~(\mref{it:mid}) and Condition~(\mref{it:mloc}) apply when $(F,d)=(F',d')=1$, giving the unique map
$$\mphi_0:\calf_0\to U, \quad 1\mapsto 1_U.$$
Since $(1_U,1_U)\in U\times_T U$, Condition~(\mref{it:mloc}) is satisfied.

Another  instructive example to study before the inductive step is the case   $k=1$, so
$\calf_1=\{1\}\cup \{\bullet_\omega\,|\, \omega\in \Omega\}$. Since $\bullet_\omega=B_+^\omega(1)$, the only map $\mphi:\calf_1\to U$ satisfying Conditions (\mref{it:mid}) --- (\mref{it:mgraft}) is given by
$$ \mphi(1)=1_U, \quad \mphi(\bullet_\omega)=\mphi(\beta^\omega(1)) =\beta^\omega(\mphi(1))=\beta^\omega(1_U).$$

Now let $k\geq 0$ be given and assume that there is a unique map $\mphi_k: \calf_k\to U$ satisfying Conditions~$(\mref{it:mid})_k$ --- $(\mref{it:mgraft})_k$. Consider $\dforest=(F,d)\in \calf_{k+1}$. If $\dforest$ is already in $\calf_k$, then $\mphi(\dforest)$ is uniquely defined by the induction hypothesis. If $\dforest\in \calf_{k+1}$ is not in $\calf_k$, then the degree of $\dforest$ is at least $1$. So Lemma~\mref{lem:fdecomp} shows that either there is a factorisation $\dforest=\dforest_1\cdots \dforest_n, n\geq 2,$ into independent properly decorated rooted trees, or
$\dforest=B_+^\omega(\overline{\dforest})$ for $\overline{\dforest}$ independent of $\bullet_\omega$ and necessarily in $\calf_k$. The assignment
\begin{equation}
\mphi_{k+1}(\dforest)=\left\{\begin{array}{ll} \mphi_k(\dforest_1)\cdots \mphi_k(\dforest_n), & \text{ if } \dforest = \dforest_1\cdots \dforest_n, \\
\beta^\omega(\mphi_k(\overline{\dforest})), & \text{ if } \dforest = B_+^\omega(\overline{\dforest}). \end{array} \right.
\mlabel{eq:mphik+1}
\end{equation}
is then well-defined  since $\mphi_{k}$ satisfies Conditions $(\mref{it:mid})_{k}$ --- $(\mref{it:mgraft})_{k}$ by assumption.

Note that this is in fact the only way to define $\mphi_{k+1}$ satisfying the conditions in Lemma~\mref{lem:char}, proving the uniqueness of $\mphi_{k+1}(\dforest)$.

Next we verify that the map $\mphi_{k+1}$ obtained this way indeed satisfies Conditions $(\mref{it:mid})_{k+1}$ --- $(\mref{it:mgraft})_{k+1}$. By the above equation and the inductive hypothesis, $\mphi_{k+1}$ satisfies Conditions~$(\mref{it:mid})_{k+1}$, $(\mref{it:mprod})_{k+1}$ and $(\mref{it:mgraft})_{k+1}$.

To verify Condition~$(\mref{it:mloc})_{k+1}$, consider $\dforest, \dforest'\in \calf_{k+1}$ with $\dforest \top_{\calf_{\Omega,\top_\Omega}} \dforest'$.
Depending on whether $\dforest$ or $\dforest'$ lies or not in $\calf_k$, there are four cases to consider.
In the case when both $\dforest$ and $\dforest'$ are in $\calf_k$, the condition is satisfied by the induction hypothesis. For the remaining three cases, the verifications are similar, the most complicated one being when neither $\dforest$ nor $\dforest'$ lies in $\calf_k$. So we will only verify this case. For this case, we further have four subcases depending on which of the two forms in Lemma~\mref{lem:fdecomp} that $\dforest$ or $\dforest'$ takes.
\smallskip

\noindent
{\bf Subcase 1. $\dforest=\dforest_1\cdots \dforest_n, n\geq 1,$ for independent properly decorated trees $\dforest_1,\cdots, \dforest_n$ and $\dforest'=\dforest_1'\cdots \dforest'_{n'}$ for independent properly decorated trees $\dforest'_1,\cdots, \dforest'_{n'}$.} Then all the factor trees are pairwise independent. Since all the factor trees are in $\calf_k$, by the inductive assumption of Condition~$(\mref{it:mloc})_k$, their images
$$\mphi_k(\dforest_i), 1\leq i\leq n, \quad \mphi_k(\dforest'_j), 1\leq j\leq n',$$ are pairwise independent in $U$. By Lemma~\mref{lem:locprod}, the products
$\mphi_k(\dforest_1)\cdots \mphi_k(\dforest_n)$ and $\mphi_k(\dforest'_1)\cdots \mphi_k(\dforest'_{n'})$ are independent. But by the construction of $\mphi_{k+1}$ in Eq.~(\mref{eq:mphik+1}), the last two products equal to $\mphi_{k+1}(\dforest_1\cdots \dforest_n)$ and $\mphi_{k+1}(\dforest'_1\cdots \dforest'_{n'})$. This gives Condition~$(\mref{it:mloc})_{k+1}$ in this subcase.
\smallskip

\noindent
{\bf Subcase 2.  $\dforest=B_+^\omega(\overline{\dforest})$ for $(\omega,\overline{\dforest})\in \Omega\times_\top \calf_{\Omega,\top_\Omega}$ and $\dforest'=\dforest'_1\cdots \dforest'_n, n\geq 2,$ for independent properly decorated trees $\dforest'_1,\cdots, \dforest'_n$.} Since $n\geq 2$, we have $\dforest'=\dforest'_A \dforest'_B$ with independent $\dforest'_A$ and $\dforest'_B$, both in $\calf_k$. Since $\dforest$ and $\dforest'$ are independent, we have
$$(\omega,\overline{\dforest},\dforest'_A,\dforest'_B)\in \Omega\times_\top \calf_k\times_\top \calf_k \times_\top \calf_k.$$
Then Conditions~$(\mref{it:mloc})_k$ and $(\mref{it:mact})_k$ lead to
$$(\omega,\mphi_k(\overline{\dforest}),\mphi_k(\dforest'_A),\mphi_k(\dforest'_B)) \in \Omega\times_\top U\times_\top U \times_\top U.$$
Applying Lemma~\mref{lem:itact} gives $(\beta^\omega(\mphi_k(\overline{\dforest})), \mphi_k(\dforest'_A)\mphi_k(\dforest'_B)) \in U \times_\top U$, that is,
$(\mphi_k(\beta^\omega(\overline{\dforest})), \mphi_k(\dforest'_A\dforest'_B))$ is in $U \times_\top U$, by Eq.~(\mref{eq:mphik+1}). This gives Condition~$(\mref{it:mloc})_{k+1}$ in this subcase.
\smallskip

\noindent
{\bf Subcase 3. $\dforest=\dforest_1\cdots \dforest_n, n\geq 2,$ for independent properly decorated trees $\dforest_1,\cdots, \dforest_n$ and $\dforest'=B_+^\omega(\overline{\dforest}')$ for $(\omega,\overline{\dforest}')\in \Omega\times_\top \calf_{\Omega,\top_\Omega}$.} This subcase follows from the previous subcase by the commutativity of the concatenation of the forest product and the \loc relation.
\smallskip

\noindent
{\bf Subcase 4. $\dforest=B_+^\omega(\overline{\dforest})$ and $\dforest'=B_+^{\omega'}(\overline{\dforest}')$ for $(\omega,\overline{\dforest}), (\omega',\overline{\dforest}')\in \Omega\times_\top \calf_{\Omega,\top_\Omega}$.}
Since the two forests are independent, we have
$$(\omega,\overline{\dforest},\omega',\overline{\dforest}')\in
\Omega\times_\top \calf_{k}\times_\top \Omega \times_\top \calf_k.$$
Then the locality of $\mphi_k$, in particular Condition~$(\mref{it:mact})_k$,  guaranteed by the induction hypothesis, gives
$$(\omega,\mphi(\overline{\dforest}),\omega',\mphi(\overline{\dforest}'))\in
\Omega\times_\top \calf_{k}\times_\top \Omega \times_\top \calf_k,$$
which yields, by Lemma~\mref{lem:itact} and Eq.~(\mref{eq:mphik+1}),
$$(\mphi(B_+^\omega(\overline{\dforest})), \mphi(B_+^{\omega'}(\overline{\dforest}'))) =(\beta^\omega(\mphi(\overline{\dforest})), \beta^{\omega'}(\mphi(\overline{\dforest}')))\in
U\times_\top U.$$
This gives Condition~$(\mref{it:mloc})_{k+1}$ in this subcase.

We have therefore completed the verification of Condition~$(\mref{it:mloc})_{k+1}$.

Let us finally check Condition~$(\mref{it:mact})_{k+1}$ assuming the induction hypothesis, distinguishing two cases:  $\dforest\in \calf_{k+1}$ is of the form $\dforest=\dforest_1\cdots \dforest_n, n\geq 2,$ for independent properly decorated trees $\dforest_1,\cdots, \dforest_n\in \calf_k$ or $\dforest=B_+^{\omega'}(\overline{\dforest})$ for $(\omega',\overline{\dforest})\in \Omega\times_\top \calf_k$.

In the first case, we write $\dforest=\dforest_A \dforest_B$ with $\dforest_A, \dforest_B\in \calf_k$. Then from $(\omega,\dforest)\in \Omega\times_\top \calf_{k+1}$ we have
$(\omega,\dforest_A,\dforest_B)\in \Omega\times_\top \calf_k\times_\top \calf_k$ which implies
$(\omega,\mphi_k(\dforest_A),\mphi_k(\dforest_B))\in \Omega\times_\top U\times_\top U$. This gives
$$ (\omega,\mphi_{k+1}(\dforest_A\dforest_B)) =(\omega,\mphi(\dforest_A)\mphi(\dforest_B))\in \Omega\times_\top U,$$ as needed.

In the second case, similarly we have
$(\omega,\omega',\overline{\dforest}')\in \Omega\times_\top \Omega\times_\top \calf_k$ which implies
$(\omega,\omega',\mphi(\overline{\dforest}'))\in \Omega\times_\top \Omega\times_\top U$.
Therefore,
$(\omega,\mphi(B_+^{\omega'}(\overline{\dforest}'))) =(\omega,\beta^{\omega'}(\mphi(\overline{\dforest}')))$
is in $\Omega\times_\top U$ by the assumption on $\beta$.

This completes the verification of Condition~$(\mref{it:mact})_{k+1}$. Together with the verification of the other conditions for the existence of $\mphi_{k+1}$ above, as well as that of the uniqueness of $\mphi_{k+1}$ after Eq.~(\ref{eq:mphik+1}), the inductive step is completed.
\end{proof}

\subsection{Branching of \loc maps}

We next move to the relative case. Given two   \loc  sets   $(\Omega_i, \top _{\Omega_i})$,  $i=1,2 $, let
\begin{itemize}
  	\item
  	$\mathcal{L}_{\top }( \Omega_1, \Omega_2)$ denote the set of   \loc maps  $\phi:\Omega_1\longrightarrow \Omega_2$;
  	\item  $\mathcal{L}_{\Omega_1,\Omega_2}(U_1, U_2)$ denote the set of morphisms from a $(\Omega_1,\top_{\Omega_1})$-operated \loc structure $(U_1,T_{U_1},\beta_1)$ of certain type (i.e., semigroup, monoid, nonunitary algebra, algebra) to a $(\Omega_2, \top _{\Omega_2})$-operated \loc structures $(U_2, \top_{U_2}, \beta_2)$ of the same type.
  	\end{itemize}

Recall that all these sets  are equipped with the independence relation of maps: $\phi, \psi: (A,\top _A)\to (B,\top _B)$
  		\[\phi\top \psi\Longleftrightarrow \left(a_1\top _A a_2\Longrightarrow \phi(a_1)\top _B \psi(a_2)\right).\]

By the same proof as the one of  Proposition~\mref{pp:liftrel}, we obtain

\begin{prop}\label{pp:liftedphi} Let
$\phi: \left(\Omega_1,\top _{\Omega_1}\right)\longrightarrow \left(\Omega_2,\top _{\Omega_2}\right)$ be a \loc map between \loc sets $(\Omega_1,\top_{\Omega_1})$ and $(\Omega_2,\top_{\Omega_2})$.

\begin {enumerate}
\item
For any commutative $\left(\Omega_2,\top _{\Omega_2}\right)$-operated \loc  monoid  $(U, \top _U ,\beta _{U,+} ,m_U, 1_U)$ $($with its $(\Omega_1,\top_{\Omega_1})$-operated \loc monoid structure induced by $\phi$$)$, there is a unique lift of $\phi$ to a  morphism  of operated  commutative
\loc monoids $\phi^\sharp: {\bbf }_{\Omega_1, \top _{\Omega_1}}\longrightarrow U$, which gives rise to a map
$$	\sharp: \left(	\mathcal{L} _\top( \Omega_1, \Omega_2),\top\right) \longrightarrow  \left(\mathcal{L}_{\Omega_1,\Omega_2}(\bbf _{\Omega_1,\top _{\Omega_1}}, U),\top\right), \quad
					\phi\longmapsto  \phi^\sharp.
$$
\item For any commutative \loc  algebra  $(U, \top _U ,\beta _{U,+} ,m_U, 1_U)$, there is a unique lift of $\phi$ to a morphism  of operated  commutative \loc algebras   $\phi^\sharp: {\mathcal F}_{\Omega_1, \top _{\Omega_1}}\longrightarrow U$, which gives rise to a map
$$
\sharp: \left(	\mathcal{L} _\top( \Omega_1, \Omega_2),\top\right) \longrightarrow \left(\mathcal{L}_{\Omega_1,\Omega_2}(\bfk\,\calf _{\Omega_1,\top _{\Omega_1}}, U),\top\right), \quad
					\phi\longmapsto  \phi^\sharp.
$$
\end{enumerate}
				We call $\phi^\sharp$ the {\bf $\phi$-lifted   map}, which by construction is characterised by the following properties:
				  \begin{eqnarray}
					 & \phi^\sharp(\emptyset) = 1_{U}, \label{eq:identity} \\
					 &    \phi^\sharp ((F_1,d_1)\cdots (F_n,d_n)) = \phi^\sharp(F_1,d_1)\cdots\phi^\sharp(F_n,d_n),\label{rec_branching_proc}\\
					 &{ \phi^\sharp \left( B_+^\omega(F,d)\right) =  \beta_{U,+}^{\phi(\omega)}\left(\phi^\sharp(F,d)\right),}  \label{concatenation}
					 \end{eqnarray}
					 for any mutually independent properly $\Omega_1$-decorated rooted forests $(F_1,d_1),\cdots, (F_n, d_n) $    and any $\omega\in \Omega_1$ independent of $(F,d)$.
\end{prop}

The following example of \loc operated monoids will be useful in the sequel.

\begin{lem}\label{lem:Omegaphi}
 Let $(\Omega,\top_\Omega,m_\Omega,1_\Omega)$ be a \loc monoid and let $\phi:\Omega\longrightarrow\Omega$ be a map such that $\phi \top \Id _\Omega$. Define
$$
\beta_\phi:\top _\Omega  \longrightarrow\Omega, \quad
	 (\omega,\omega')  \longmapsto \beta_\phi^\omega( \omega'):=\phi(m_\Omega (\omega,\omega')).
$$
Then $(\Omega,\top_\Omega,\beta _\phi, m_\Omega,1_\Omega)$
is an {$(\Omega,\top _\Omega)$}-operated monoid.
\end{lem}

\begin{proof}
Let $\left(\Omega,\top _{\Omega}\right) $ and $\phi:\Omega\longrightarrow\Omega$ be as in the statement of the lemma. By symmetry of $\top_\Omega$, $\phi\top Id$ implies $\phi\top\phi$, i.e. that $\phi$ is a \loc map.

The axioms for an $(\Omega,\top_\Omega)$-\loc operated monoid on $U=\Omega$ are checked as follows.

Let $\left( \omega, u, u' \right)$ be in $\Omega^{_\top 3}$. Then $m_\Omega (\omega,\, u)\top_\Omega  u'$ holds.
Since $\phi \top \Id _\Omega$, we obtain
$(\beta_\phi^\omega(u), u')= (\phi(m_\Omega (\omega , u)), u')$ which is in $\top_\Omega.$

Let $  \left( \omega,\omega', u \right)$ be in $\Omega^{_\top 3}$. Then $\omega \top_\Omega m_\Omega(\omega',u)$ and hence $\omega\top_\Omega \phi(m_\Omega (\omega' , u))$. This means that $(\omega, \beta_\phi^{\omega'}(u))$ is in $\top _\Omega$.

Let $\left( \omega, u, u' \right)$ be in $\Omega^{_\top 3}$. Then $\omega \top_\Omega  m_\Omega (u,u')$ holds, that is,
$(\omega, m_\Omega (u,u'))\in \top _\Omega$.

Finally $\Omega \times 1_\Omega \subset \top _\Omega$ since $(\Omega, \top_\Omega,m_\Omega, 1_\Omega)$ is a \loc monoid.
\end{proof}

Applying Proposition~\ref{pp:liftedphi}, we obtain

\begin{coro}\label{thm:existencemapbranching}
Let     $\left(\Omega,\top _{\Omega}\right) $ be a commutative \loc  monoid $($resp. a unital  commutative \loc  algebra$)$.
A {map}  $$\phi: \left(\Omega,\top _{\Omega}\right)\longrightarrow \left(\Omega,\top _{\Omega}\right)$$
such that $\phi \top \Id _\Omega$ induces a unique morphism of commutative \loc monoids $($resp. unital  commutative \loc  algebras$)$
   \[   \widehat \phi: \bbf _{\Omega, \top _{\Omega}}\longrightarrow (\Omega, \top_\Omega), \]
   $($resp.
  $\widehat \phi: \bfk\,\calf _{\Omega, \top _{\Omega}}\longrightarrow (\Omega, \top_\Omega)$$)$.
$\widehat \phi$ is called  the {\bf $\phi$-branched map}.
 By construction it is characterised by the following properties:
\begin{eqnarray}
                    &\hat  \phi(\emptyset) = 1_{\Omega} \label{eq:identitybranched} \\
                    &  \hat   \phi  ((F_1,d_1)\cdots (F_n,d_n)) = \hat \phi(F_1,d_1)\cdots\hat\phi (F_n,d_n)\label{rec_branching_procbranched}\\
                    &    \widehat \phi \big( B_+^\omega(F,d)\big) =  \phi\big(\omega\,\big(\widehat\phi (F,d)\big)\big), \label{Bbranched}
\end{eqnarray}
for any  mutually independent properly decorated forests $(F_1,d_1), \cdots, (F_n,d_n)\in \mathcal{F}_{\Omega_1,\top _{\Omega_1}}$, and any $\omega\in \Omega_1$  which is independent of $(F,d)$.
\end{coro}
 \begin{proof}
As before, $\phi\top Id$ implies $\phi\top\phi$,i.e. that $\phi$ is a \loc map.
Then by Lemma \ref{lem:Omegaphi} we can apply Proposition~\ref{pp:liftedphi} in the special case $\Omega_1=\Omega_2=\Omega$ with the map $Id_\Omega:(\Omega,\top_\Omega)\longrightarrow(\Omega,\beta_\phi)$. Then we have $\widehat\phi:=Id^\sharp$, which has the stated properties by the definition of $\sharp$.
 \end{proof}
 \begin{rk}
  \begin{itemize}
   \item Writing $\widehat\phi:=Id^\sharp$ might seem confusing since no $\phi$ appears on the right hand side. Of course, this is a notation artifact: $\phi$ does play a role in the right hand side, since the lift $\sharp$ is
   made with respect to the operation $\beta_\phi$. A more rigorous notation would be $Id^{\sharp,\beta_\phi}$ which we have not opted for in order to lighten the notations.
   \item One could also prove the Corollary as a consequence of Theorem~\mref{thm:univ_prop_trees_loc}, in the same way that Proposition~\mref{prop:lift_phi_non_loc} is derived from
   Corollary~\mref{co:initop}.
  \end{itemize}
 \end{rk}

\section {Multivariate regularisation of branched integrals}
\mlabel{sec:tree}

We apply the framework previously developed to study Kreimer's toy model in building a regularisation by  means of locality algebra homomorphisms. So our first goal is to define a suitable decoration set. There are several possible decoration sets; in this paper we use a very simple one  similar to that in~\cite {GPZ4}. An alternative choice can be found in~\cite{CGPZ2}.

\subsection{Linear complex powers: the \loc algebra $\calm [\call]$}
\mlabel{ss:ComplexPower}

We adapt the terminology from~\cite{GPZ1} to the algebraic locality framework developed in~\cite{CGPZ1}.

Consider a sequence $(\R^\infty,Q):=(\bbR^n,Q_n)_{n\geq 1}$ of Euclidean space with inner product $Q_n$ on $\bbR^n$ such that, under the direct system $i_n: \R ^n \to \R ^{n+1}, n\geq 1,$ of standard embeddings, we have
$Q_{n+1}|_{\R^n\times \R^n}=Q_n, n\geq 1.$

Since $Q_n$ on $\R^n$ induces an isomorphism
$$Q_n^*: \R ^n \to (\R ^n)^*,
$$
the maps
$$j_{n}=(Q_n^*)^{-1} \circ i_n^*\circ Q _{n+1}^*: \R ^{n+1} \to \R ^n,
$$
give rise to the direct systems
$$j_{n}^*: (\R ^n)^*\to (\R ^{n+1})^*,
$$
$$j_{n}^*: \calm (\R ^n \otimes \C)\to \calm (\R ^{n+1}\otimes \C).
$$
Here $\calm(\R^n \ot \C)$ is the algebra of multivariate meromorphic germs at $0$ with linear poles and real coefficients ~\mcite{GPZ1,GPZ2}.

We set
$$\call: =\varinjlim _n (\bbR ^n)^*,
\quad
\calm:=\varinjlim _n \calm (\bbR ^n\otimes \bbC).
$$

By~\mcite{GPZ2,GPZ3}, the algebra $\calm(\R^n\ot \C)$ is equipped with the linear decomposition
$$ \calm(\R^n\ot \C) = \calm_+(\R^n\ot \C)\oplus \calm ^Q_-(\R^n\ot \C)$$
as the direct sum of the subalgebra $\calm_+(\R^n\ot \C)$ of holomorphic germs and the space $\calm_-(\R^n\ot \C)$ generated by {\bf polar germs}, defined to be fractions of the form
\begin{equation}
\frac{h(\ell_1,\cdots, \ell_k)}{L_1\cdots L_n},
\mlabel{eq:polar}
\end{equation}
where $\ell_1,\cdots, \ell_k, L_1,\cdots,L_n$ are linear form such that $\{\ell_1,\cdots,\ell_k\}\perp\{L_1,\cdots,L_n\}$.
These decompositions are compatible with the  maps $j_{n}^*$, thus we have
\begin{equation}
\calm =\calm_+\oplus \calm_-, \quad \calm_\pm:= \varinjlim_n \calm_\pm (\R^n\ot \C).
\mlabel{eq:calmdecom}
\end{equation}
Let
\begin{equation}
\pi_+: \calm \longrightarrow \calm_+
\mlabel{eq:pi+}
\end{equation}
denote the projection onto $\calm_+$ along $\calm_-$.

In the algebra of functions with complex variables in $\C^\infty$ and {with} one real variable $x$, we consider the $\calm$-module of linear combinations
$$\Big\{ \sum_{i=1}^k f_i x^{L_i} \,\Big|\, f_i \in \calm, L_i\in \call, 1\leq i\leq k, k\geq 1\Big\}.$$
It is an $\calm$-subalgebra since $x^L, L\in \call$ is closed under multiplication. It can further be checked that $x^{L}, L\in \call,$ are linearly independent over $\calm$. Thus it is  isomorphic to the group ring $\calm[\call]$ over $\calm$ generated by the additive monoid $\call$. We will   henceforth make this identification.

\begin{rk} Elements $x^L \in \calm[\call]$ are   particular instances of the more general (multivariate) holomorphic families of classical symbols considered in   \cite{CGPZ2} -- with the difference that in \cite{CGPZ2} they are
symbols on $[0, +\infty)$ whereas  here  they are defined on $(0, +\infty)$.
\end{rk}

We equip the algebra $\calm$ with the
independence relation $\perp$  introduced in \cite[Definition 2.14]{CGPZ1}
\begin{equation}\label{locperp}\left(f\perp g\right) \ \Longleftrightarrow \ \left( {\rm Dep}(f) \perp {\rm Dep} (g)\right).\end{equation}
The pointwise product gives rise to a \loc  algebra $\left(\calm , \perp\right)$ and $(\call,\perp)$ is viewed as a \loc subspace of the \loc linear space  $\left(\calm , \perp\right)$. See~\mcite{CGPZ3} for another \loc structure on $\calm$.

The \loc relation $\perp$ on $\calm$ induces one on the space $\calm [\call] $:
\begin{equation}
\left(\sum_i f_ix^{L_i}\perp \sum_j g_jx^{\ell _j}\right) \ \Longleftrightarrow \ \left( \{f_i, L_i\}_i\perp \{g_j, \ell _j\}_j \text{ in } \calm\right)
\mlabel{eq:mrel}
\end{equation}
and  $\left(\calm [\call ], \perp\right)$ is a locality algebra.

Using $\call $ as the set of decorations for rooted forests in the previous sections, we have the $\call $-operated monoid   $\calf_{\call }$, the $(\call,\perp)$-operated \loc monoid $\calf_{\call,\perp}$ and the corresponding locality algebra $\R \calf_{\call,\perp}$.

\subsection{$\calm [\call]$ as an operated \loc algebra}
\label{subsec:localg}
We revisit a map defined  in  \cite {GPZ4}, viewed here as an operating map on $\calm [\call]$:

\begin {lem}\label{lem:opL} For any $L\in \call$, the operator
$$ {\mathcal I}^L(f)(x):=\int_0^\infty  \frac {f(y)y^{-L}}{y+x} dy,$$
defines a linear map from $\calm [\call]$ to $\calm [\call]$.
With
$$
 {\mathcal I}: \call\times \calm [\call]{\longrightarrow} \calm [\call],\quad
(L, f)\longmapsto {\mathcal I}^L(f),
$$
denoting the resulting action of $\call$ on $\calm[\call]$, the triple $(\calm [\call], \perp, {\mathcal I})$ is an  $(\call,\perp)$-operated algebra.
\end{lem}

\begin {proof} We first prove that ${\mathcal I}^L$ defines a map from $\calm [\call]$ to $\calm [\call]$. By the $\calm$-linearity, we only need to prove $\cali ^L(x^\ell )\in \calm [\call]$.

For fixed $x>0$ and $z\in \C$, the map $y\longmapsto \frac{y^{-z}}{(y+x)^k}$ is locally integrable on $ (0, +\infty)$. Since $\frac{y^{-z}}{(y+x)^k}\underset{y\to 0}{\sim}C\, y^{-z}$ and $\frac{y^{-z}}{(y+x)^k} \underset{y\to \infty}{\sim} y^{-z-k}$, the integral $\int_0^\infty \frac{y^{-z}}{(y+x)^k}\,dy$ converges absolutely on the strip $ \Re (z) \in (1-k,1)$ and  the map $x\longmapsto \int_0^\infty \frac{y^{-z}}{y+x}\, dy$ is smooth with $k$-th derivative $x\longmapsto (-1)^k\,k!\,\int_0^\infty \frac{y^{-z}}{(y+x)^{k+1}} \, dy. $

For $\alpha >0$, we have
$$\lim _{x\to 0+}x^\alpha \,\ln (x)=0 , \ \lim _{x\to \infty }x^{-\alpha} \,\ln (x)=0,
$$ so that for fixed $x>0$,  the map $z\longmapsto \int_0^\infty \frac{y^{-z}}{y+x}\,dy$ is
  holomorphic in $z$ on the  strip $\Re (z)\in (0,1)$ with derivative $z\longmapsto -z\,\int_0^\infty \frac{\ln y\, y^{-z}}{y+x}\,dy$.
For any real number $ a \in (0, 1)$,  an explicit computation~\cite[Lemma 4.5]{GPZ4} gives:
 $$\int_0^\infty \frac{x^{-a}\,d x}{x+1}
  = \frac{\pi}{\sin (\pi a)},$$
 and hence
$$\int_0^\infty \frac{y^{-a}}{y+x}\,dy=
 x^{-a}\, \int_0^\infty \frac{y^{-a}\,d y}{y+1}=\frac{\pi}{\sin(\pi a)} x^{-a}.$$  It follows that for any complex number in the strip $\Re (z)\in (0,1)$, we have
$$\int_0^\infty \frac{y^{-z}}{y+x}\,dy=\frac{\pi}{\sin(\pi z)} x^{-z}.
$$

As a consequence of  this explicit formula, the map $z\longmapsto  \int_0^\infty \frac{y^{-z}}{y+x}\,dy$ extends to a meromorphic family (in the  parameter $z$) of smooth functions in $x$. Thus, when restricted to a neighborhood of $0$ in the parameter space, ${\mathcal I}^L$ can be viewed as a $\calm$-linear map from $\calm [\call] $ to
$\calm [\call] $. More precisely,
\begin{equation} \label{eq:formula_IL}
 {\mathcal I}^L(f x^{-\ell})=\left(x\mapsto f(x) \frac{\pi}{\sin(\pi(L+\ell))} x^{-L-\ell}\right), \quad f\in \calm, L\in \call
\end{equation}
(on the l.h.s, we write $x^{-\ell}$ for $x\to x^{-\ell}\in\calm[\call]$).
The conditions for $(\calm [\call],{\mathcal I} )$ to be an $(\call,\perp)$-operated locality algebra are then easy to check.
\end{proof}

\subsection{Locality morphisms on rooted forests}
\label{subsec:locmorforests}
Recall from Eq.~(\mref{eq:mrel}) that $\R \calf_ {\call, \perp}$ comes equipped with a \loc relation induced by the \loc relation $\perp$ on $\call$, for which $\left(\R \calf_ {\call, \perp}, \perp\right)$
is a \loc algebra. The universal property of the $(\call,\perp)$-operated algebra $\R \calf_{\call,\perp}$ discussed in Theorem~\mref{thm:locinit} yields a unique locality algebra homomorphism
$${\mathcal R} : (\R  \calf_{\call,\perp }, \top_{\calf_{\call,\perp }},B_+)\to (\calm  [\call], \top,{\mathcal I}),
$$
(with $\mathcal{I}:=\{\mathcal{I}^L,L\in\mathcal{L}\}$)
which is characterized by the following conditions:
\begin{align}
\calr(\bullet_ L) & = \int_0^\infty \frac{y^{-L}}{y+x}dy = \frac{\pi}{\sin (\pi L)} x^{-L},
\mlabel{eq:rinit}
\\
\calr((F_1,d_1)(F_2,d_2)) & = \calr(F_1,d_1) \calr(F_1,d_2) \quad \text{for all } (F_1,d_1)(F_2,d_2)\in \calf_{\call,\perp},
\label{eq:rprod}\\
\calr\left( B_+^L((F, d))\right) & = {\mathcal I}^L\left(\calr((F,d))\right) \quad \text{for all } B_+^L((F, d))\in \calf_{\call,\perp}.
\mlabel{eq:rgraft}
\end{align}

The subsequent statement follows from~\cite[Lemma 4.5]{GPZ4}.
\begin{lem} \label{lem:RFd}
For any properly decorated rooted forest $(F,d)\in  \calf_{\call,\perp }$,
$${\mathcal R}  (F,d)(x)=x^{-\sum _{v\in V(F)}d(v)}\,\prod_{v\in V(F)}\frac {\pi}{\sin {(\pi L_v)}},
$$
where, for any $v\in V(F)$, $L_v$ is the sum of decorations associated to vertices of the sub-trees with root $v$: $L_v:=\sum_{w\geq v} d(w)$.
\mlabel{lem:Rformula}
\end{lem}

\begin{proof}
Because of the uniqueness of the operated algebra homomorphism $\calr$ satisfying the conditions in Eqs.~(\mref{eq:rinit})--(\mref{eq:rgraft}), we only need to verify that the desired equation in
the lemma satisfies the same three properties. Taking $(F,d)=\bullet_L$ gives Eq.~(\mref{eq:rinit}). The second and third properties are easy to verify noting that the right hand side of the desired equation is
$$\prod_{v\in V(F)}x^{-d(v)}\frac {\pi}{\sin (\pi L_v)}
$$
and using the explicit expression \eqref{eq:formula_IL} of $\mathcal{I}^L$.
\end{proof}

We also prove the following property for later use.
\begin{lem}
With the notation in Lemma~\mref{lem:Rformula}, the set $\{L_{v}\}_{v\in V(F)}$ is linearly independent.
\mlabel{le:indep}
\end{lem}
\begin{proof}
Suppose not, then there is a linear combination
$\sum_{v\in V(F)} c_v L_{v}=0$ in which not all coefficients are zero. Let $v_0$ be maximal such that $c_{v_0}\neq 0$ with respect to the order $v\leq w$ if $w\in F_v$. By the definition of $L_{v}$, we have
$$\sum_{v\in V(F)} c_v L_{v}=\sum_{w\in V(F)} b_w d(w),$$
where
$b_w=\sum_{v\leq w} c_{v}$. Since $F$ is properly decorated, from $\sum_{v\in V(F)} c_v L_{v}=0$ we obtain $b_w=0$ for all $w\in V(F)$.
Now by the maximality of $v_0$ with nonzero coefficients, we have
$0=b_{v_0}=\sum_{v\leq v_0} c_v=c_{v_0}$. This gives the desired contradiction.
\end{proof}

\section{Multivariate renormalisation of branched integrals}
\label{sec:renCK}
\subsection{The minimal substraction scheme on fractions}
In preparation for the actual multivariate renormalisation, in this paragraph, we  investigate the holomorphic projections of a special family of meromorphic germs, that is $\pi_+\left( \frac{g(L_{w}, w\in W )}{\prod_{v\in V}L_{v}}\right)$ of a fraction $ \frac{g(L_{w}, w\in W )}{\prod_{v\in V}L_{v}}$ built from a holomorphic germ $g$, together with  their evaluation  at zero ${\rm ev}_0\circ \pi_+\left( \frac{g(L_{w}, w\in W )}{\prod_{v\in V}L_{v}}\right)$, for a fixed set of   independent linear forms $\{L_w, w\in W\}$ indexed by $W$ and a subset $V$ of $W$. The latter expression amounts to implementing a minimal substraction scheme on the fraction  $\frac{g(L_{w}, w\in W )}{\prod_{v\in V}L_{v}}$.
\begin{thm}\label{thm:reductionproc} Given  $V\subset W$ and $g=g(z_w, w\in W)$    a holomorphic germ at zero,
\begin{enumerate}
\item there are holomorphic germs at zero $f_v(L_{w}, w\in W)$ for $ v\in V$   such that
 \begin{equation}\label{eq:piplusg}\pi_+\left( \frac{g(L_{w}, w\in W )}{\prod_{v\in V}L_{v}}\right)=\sum_{v\in V}\pi_+\left( \frac {f_v(L_{w}, w\in W)}{\prod_{u\in V\setminus \{v\}}L_u} \right).\end{equation}
\mlabel{it:red1}
 \item for two sets of linear forms $L_{iv}, v\in W, i=1,2$ related by  \begin{equation}
Q(L_{1v},L_{1w})=c\, Q(L_{2v}, L_{2w})\quad \text{for all } v, w\in W,
\mlabel{eq:qc}
\end{equation}
with $c$ some constant, we have \begin{equation}
{\rm ev}_0 \circ \pi _+  \left (\frac {g(L_{1w}, w \in W)}{\prod _{v\in V}L_{1v}} \right )={\rm ev}_0 \circ \pi _+  \left (\frac {g(L_{2w}, w\in W)}{\prod _{v\in V}L_{2v}} \right ).
\mlabel{eq:veq}
\end{equation}
\label{it:red2}
\end{enumerate}
\end{thm}
\begin{proof} The proof uses techniques similar to those of the proof of \cite[Theorem 2.11]{GPZ3}.
(\ref{it:red1}).
To prove Eq.~(\ref{eq:piplusg}), we proceed by induction on $|V|$.

When $|V|=0$, the fraction is already holomorphic and hence the result is the holomorphic function. In general, for a given $V\subseteq W$, let
$$ (\R^n)^* = {\rm span}(V)\bigoplus {\rm span} (V)^{\perp} $$
be the decomposition of $(\R^n)^*$ as the direct sum of the linear span of $L_v, v\in V$ (which, without loss of generality, can be assumed to be linearly independent) and its orthogonal complement.
Then $w\in W$ has the decomposition
\begin{equation}
L_{w}=L''_w+L'_w=\sum _{v\in V} a_{wv}L_{v}+L'_w.
\mlabel{eq:perp}
\end{equation}
The projections $L'_w$ and $L''_w$ are unique, thus the choice of the coefficients $a_{wv}$ with respect to the linearly independent set $\{L_v, v\in V\}$ is unique.

We therefore get the decomposition
$$
\frac {g(L_{w}, w\in W)}{\prod _{v\in V}L_{v}}=\frac {g(L'_{w}, w\in W)}{\prod _{v\in V}L_{v}}+\frac {g(L_{w},w \in W)-g(L'_{w}, w \in W)}{\prod _{v\in V}L_{v}},$$
where
$$\frac {g(L'_{w}, w\in W)}{\prod _{v\in V}L_{v}}$$
is a polar germ as defined in Eq.~(\mref{eq:polar}) and hence is annihilated by $\pi_+$. Thus we only need to consider the second fraction
\begin{eqnarray}
&&\frac {g(L_{w},w \in W)-g(L'_{w}, w \in W)}{\prod _{v\in V}L_{v}} \notag\\
 &=& \frac {g(\sum_{v\in V} a_{wv}L_{v}+L'_{w},w \in W)-g(L'_{w}, w \in W)}{\prod _{v\in V}L_{v}}.
\mlabel{eq:second}
\end{eqnarray}
Using the expression of $L'_w=L_w-\sum a_{wv}L_v$ from Eq.~(\mref {eq:perp}), we note  that the fraction is a function in $L_{v}, v\in W$.

Before proceeding further, we introduce some notations to simplify the presentation, setting
$$ \vec{L}:=(L_{w}, w\in W), \quad \vec{L}':=(L'_{w}, w\in W).$$
For $v\in V$, $u\in W$, we set $\vec{L}_{uv}:=(\delta_{wu} L_{v}, w\in W)$. Then we have
\begin{equation}
\vec{L}=\vec{L}'+\sum_{j=1}^N c_{j} \vec{M}_{j}, \quad \text{where }\vec{M}_{j}=\vec{L}_{uv} \text{ and } c_{j}=a_{uv} \text{ for some }v\in V, u\in W.
\mlabel{eq:perp2}
\end{equation}

With these notations, we have the telescopic expansion

\begin{eqnarray*}
&&g\big(L'_{w}+\sum_{v\in V} a_{wv}L_{v},w \in W\big)-g(L'_{w}, w \in W)\\
&=&
g\big(\vec{L}'+\sum_{j=1}^N c_{j} \vec{M}_{j}\big)
-g(\vec{L}') \\
&=& \sum_{\ell=1}^N\Big( g\big(\vec{L}'+\sum_{j=1}^\ell c_{j}\vec{M}_{j} \big)
- g\big(\vec{L}'+\sum_{j=1}^{\ell-1} c_{j}\vec{M}_{j}\big)\Big)
\end{eqnarray*}
with the convention that the sum over the empty set is zero.

If $\vec{M}_{\ell }=\vec{L}_{u_\ell v_\ell}$, then the fraction
$$h_\ell(L_{w}, w\in W):=\frac{g\big(\vec{L}'+\sum_{j=1}^\ell c_{j}\vec{M}_{j} \big)
- g\big(\vec{L}'+\sum_{j=1}^{\ell-1} c_{j}\vec{M}_{j}\big)}{{L}_{v_\ell}}$$
is holomorphic  as a function  of $L_w, w\in W$. So
\begin{equation*} \frac{g\big(\vec{L}'+\sum_{j=1}^\ell c_{j}\vec{M}_{j} \big)
- g\big(\vec{L}'+\sum_{j=1}^{\ell-1} c_{j}\vec{M}_{j}\big)}{\prod_{v\in V} L_{v}}
=\frac {h_\ell(L_{w}, w\in W)}{\prod_{v\in V\setminus \{v_\ell\}} L_{v}}, 1\leq \ell \leq N.
\mlabel{eq:frec}
\end{equation*}
Now taking
$$f_v(L_{w}, w\in W):=\sum h_\ell(L_{w}, w\in W)$$
where the sum is over all $\ell$ with $\vec{M}_{\ell}=\vec{L}_{uv}$ for some $u\in W$, we have Eq. (\ref{eq:piplusg}).
\smallskip

\noindent
(\ref{it:red2}).
To prove Eq.~(\ref{eq:veq}) we proceed by induction on $|V|$ along the lines of the first part of the proof. The identity clearly holds when $|V|=0$ since both sides equal to $g(z_w, w \in W)|_{z_w=0}$.

Assume that the equation holds when $|V|=k$ for $k\geq 0$ and consider the case when $|V|=k+1$. For $w\in W$ and $i=1,2$, the decomposition in  Eq.~(\mref{eq:perp}) yields
\begin{equation}
L_{iw}=L''_{iw}+L'_{iw} =\sum _{v\in V} a_{iwv}L_{iv}+L'_{iw}.
\mlabel{eq:perp3}
\end{equation}

By Eq.~(\ref{eq:qc}), we have
\begin{claim}
For any $w\in W$ and $v\in V$, the equations $a_{1wv}=a_{2wv}$ holds.
\mlabel{cl:aequal}
\end{claim}

\begin{proof} Taking inner product with $L_{iu}$, $u\in V$ in Eq.~(\ref {eq:perp3}), we have a linear system for $a_{iwv}$:
$$\sum _{v\in V} Q(L_{iv}, L_{iu})a_{iwv}=Q(L_{iw}, L_{iu}).
$$
Since $\{L_{iu}, u\in V\}$ is linearly independent by assumption,
the matrix
$$(Q(L_{iv}, L_{iu})_{u,v \in V})
$$
is invertible and this system has a unique solution, as already shown before. Further, by Eq.~(\mref{eq:qc}), the coefficients of the linear system for $i=2$ are $c$ times those for $i=1$. Hence the two systems have the same solution, proving the equation.
\end{proof}

With notations similar to those of the first part of the proof,
\begin{eqnarray*}
\lefteqn{{\rm ev}_0 \circ \pi _+  \left ( \frac {g(L_{iw}, w\in W)}{\prod _{v\in V}L_{iv}}\right )={\rm ev}_0 \circ \pi _+  \left (\frac {g(L_{iw},w \in W)-g(L'_{iw}, w \in W)}{\prod _{v\in V}L_{iv}}\right )}\\
&=&\sum_{\ell=1}^N {\rm ev}_0 \circ \pi _+ \Big( \frac {g\big(\vec{L}'+\sum_{j=1}^\ell c_{j}\vec{M}_{j} \big)
- g\big(\vec{L}'+\sum_{j=1}^{\ell-1} c_{j}\vec{M}_{j}\big)}{\prod _{v\in V}L_{iv}}\Big)\\
&=&\sum_{\ell=1}^N {\rm ev}_0 \circ \pi _+ \Big( \frac {h_\ell(L_{iv}, v\in W)}{\prod _{v\in V\setminus \{v_\ell\}}L_{iv}}\Big).
\end{eqnarray*}
By Claim~\mref{cl:aequal}, if $L_{1v}$ is replaced by $L_{2v}, v\in W$, then $L'_{1v}, \vec{L}'_{1}, c_{1j}, M_{1j}, \vec{M}_{1j}$ are replaced by $L'_{2v}, \vec{L}'_{2}, c_{2j}, M_{2j}, \vec{M}_{2j}$ respectively. Thus the fraction
$$\frac {h_\ell(L_{iv}, v\in W)}{\prod _{v\in V\setminus \{v_\ell\}}L_{iv}}$$
when $i=2$ is obtained from the fraction when $i=1$ upon replacing its variables $L_{1v}$ by $L_{2v}, v\in W$.
Thus the induction assumption yields the conclusion.
\end{proof}

\subsection{Locality morphisms on properly decorated rooted forests}
\label{subsec:renmorph}

 In order to renormalise branched integrals by means of a locality multivariate regularisation, we need to introduce some locality morphisms.
We state a straightforward yet useful preliminary result without proof.

\begin{lem} The evaluation map at $x=1$:
\begin{equation}
{\rm ev}_1:={\rm ev}_{x=1}: \left(\calm [\call ], \perp\right)\longrightarrow \left(\calm, \perp\right), \quad f x^{L}\mapsto f, \quad f\in \calm, L\in \call,
\mlabel{eq:ev1}
\end{equation}
is a homomorphism of \loc algebras.
\end{lem}

Combining the locality  properties of ${\rm ev}_1$ and ${ \mathcal  R}$ leads to the following properties of their composition:
\begin{prop}
The map
$${ \mathcal  R}_1 := {\rm ev}_1 \circ { \mathcal  R}: \bbR \calf _{\call,\perp}\to \calm
$$
sends properly decorated forests to
\begin {equation}
\label {eq:BarR}
{\mathcal  R}_1(F,d)=\prod_{v\in V(F)}\frac {\pi}{\sin {(\pi L_v)}}
\end{equation}
(with $L_v$ as in Lemma \mref{lem:RFd})
and defines  a homomorphism of locality algebras
\begin{equation}\label{eq:calr1}\calr_1  : (\bbR \calf _{\call,\perp}, \perp _{\calf _{\call,\perp}})\to \left(\calm , \perp \right).\end{equation}

 In particular,  for independent properly decorated forests $(F, d)$ and $(F', d')$, we have
$$ { \mathcal  R}_1 (F, d)\perp { \mathcal  R}_1(F', d') \ \text{ and } \ { \mathcal  R}_1\left( (F, d)\bullet (F', d')\right)= { \mathcal  R}_1 (F, d)\, { \mathcal  R}_1(F', d'). $$
\mlabel{thm:locrenorm}
\end{prop}

Recall from Eq.~(\mref{eq:pi+}) that $\calm_-$ is a \loc ideal of $\calm$, implying that $\pi_+:\calm\to \calm_+$ is a \loc morphism as proved in \cite{CGPZ1}. Composing  the locality algebra homomorphism
${ \mathcal  R}_1  : \bbR \calf _{\call, \perp}\to \calm
$ with the    \loc morphism $\pi _+$   and the evaluation
$${\rm ev}_0:={\rm ev}_{z=0}:\calm_+\to \C$$
of a holomorphic germ at zero yields the subsequent statement.

\begin{propdefn}
The map ${ \mathcal  R}^{\rm ren}$  on $\bbR \calf _{\call, \perp}$   defined by
\begin{equation}
\label{eq:calrren} \mathcal  R^{\rm ren}  :={\rm ev}_0\circ\pi_+\circ  { \mathcal  R}_1
\end{equation}
   is a  locality character on the \loc algebra $\bbR \calf _{\call, \perp}$, called the {\bf renormalised character}.
\mlabel{thm:locrenormr}
\end{propdefn}

\subsection{Renormalised values on similar properly decorated rooted forests}
 \mlabel{ss:similar}
Theorem \ref {thm:reductionproc} provides a useful recursive procedure for $\calr^{\rm ren}$ by means of an algorithm to evaluate the renormalised value of any given branched integral.

Inserting
in Eq.~(\mref{eq:BarR}),  the Laurent expansion at $x=0$,
$$\frac {\pi}{\sin {(\pi x)}}=\frac{1}{x} + h(x),$$
where $h(x)$ is holomorphic, yields
\begin{eqnarray}
{\mathcal  R}_1(F,d)
&=& \prod_{ v\in V(F)}\left(\frac 1{ L_{v}}+h( L_{v})\right) \notag\\
&=&\sum _{V\subseteq V(F)}\frac 1{\prod _{v\in V}L_{v}}\left( \prod _{v\in V(F)\setminus V}h(L_{v})\right)\mlabel{eq:r1}\\
&=& \sum _{V\subseteq V(F)} \frac{g(L_{w}, w\in V(F)\setminus V)}{\prod_{v\in V}L_{v}}, \notag
\end{eqnarray}
where $g(z_w,w\in V(F)\setminus V):=\prod _{w\in V(F)\setminus V}h(z_w)$ is holomorphic. In view of Lemma~\mref{le:indep}, the linear forms $\{L_{w},w\in V(F)\}$ are linearly independent, thus the fraction can be regarded as a function with variables in $L_w, w\in V(F)$.

The  renormalised locality character ${ \mathcal  R}^{\rm ren}$ has very special properties besides functorial ones.

 \begin {defn} \label{defn:similar}Two properly decorated rooted forests $(F_1, d_1)$ and $(F_2, d_2)$ are called {\bf similar} if $F_1=F_2$ and if there exists a constant $c\in \bbR_{>0}$ such that $Q(d_1(v),d_1(v))=c\, Q(d_2(v), d_2(v))$ for any $v\in V(F_1)=V(F_2)$.
\end{defn}

\begin{coro}
\label {pp:similarity}
For similar properly decorated rooted forests $(F_1, d_1)$ and $(F_2, d_2)$, we have
$$\calr ^{\rm ren} (F_1, d_1)=\calr ^{\rm ren} (F_2, d_2).
$$
\end{coro}
\begin{proof}
Let $ F_v$ denote the (maximal) subtree
of $F=F_1=F_2$ with root $v$. As before let $L_{iv}= \sum_{w\in V(F_v)} d_i( w)$ for $i=1,2$, and $V(F)=V(F_1)=V(F_2)$.

Note that for any $v, w\in V(F)$, the intersection $F_v\cap F_w$ is either $\emptyset$ or $F_v$ (when $v\in V(F_w)$) or $F_w$ (when $w\in V(F_v)$). Thus
$$Q(L_{iv},L_{iw})=\left \{\begin{array}{ll}
0, & F_v\cap F_w=\emptyset, \\
\sum_{u\in F_v} Q(d_i(u),d_i(u)), & F_v\cap F_w = F_v,\\
\sum_{u\in F_w} Q(d_i(u),d_i(u)), & F_v\cap F_w=F_w.
\end{array} \right . $$
Thus by the similarity of $(F_1,d_1)$ and $(F_2,d_2)$, we have Eq.~(\ref{eq:qc}).

Let $F=F_i, i=1,2,$ in Eq.~(\mref{eq:r1}), we have
\begin{eqnarray*}
{\mathcal  R}_1(F_i,d_i)
&=& \sum _{V\subset V(F)} \frac{g(L_{iw}, w\in V(F))}{\prod_{v\in V}L_{iv}},
\end{eqnarray*}
where $g(z_v,v\in V(F)\setminus V):=\prod _{v\in V(F)\setminus V}h(z_v)$ is holomorphic in variables  $\{L_{iw},w\in V(F)\}$.
Theorem \ref{thm:reductionproc}.(\ref{it:red2}) with $W=V(F)$ yields the statement.
\end{proof}

This concludes the renormalisation of branched integrals via \loc morphisms by means of the multivariate renormalisation scheme developed in \cite{CGPZ1}   in the framework of  Kreimer's toy model.
\smallskip

\noindent

{\bf Acknowledgements. } The authors acknowledge supports from the Natural Science Foundation of China (Grant Nos. 11521061, 11771190, 11890663, and 11821001) and  the German Research Foundation (DFG project FOR 2402). The first  and third author thank the Perimeter Institute where the paper was completed.

\end{document}